\documentclass[11pt]{article}

\usepackage[margin=1in]{geometry}
\usepackage{amsmath, amssymb, amsthm}
\usepackage{booktabs}
\usepackage{siunitx}
\usepackage{bm}
\usepackage{graphicx}
\usepackage{float}
\usepackage{subcaption}  
\usepackage{placeins}
\usepackage{setspace}
\usepackage{natbib}  
\usepackage{hyperref}

\hypersetup{
  colorlinks=true,
  linkcolor=blue,
  citecolor=blue,
  urlcolor=blue
}
\newtheorem{theorem}{Theorem}

\newtheorem{lemma}{Lemma}
\theoremstyle{remark}
\newtheorem{remark}{Remark}

\title{\textbf{Logistic Regression Equivalent Weights for Survey Inference: Construction and Asymptotic Properties}}
\author{Seonghun Lee\\
Columbia University}

\date{August 20 2026}

\begin{document}
\maketitle

\addcontentsline{toc}{section}{Abstract}
\tableofcontents

\begin{abstract}
Equivalent weights from regression models provide a bridge between
design-based and model-based survey inference. We develop a frequentist
survey-weighting framework for logistic regression equivalent weights
under categorical poststratification. Because the logistic model-based
population estimator is nonlinear in the observed outcomes, we define
equivalent weights through a local first-order representation, with each
weight given by the scaled derivative of the population estimator with
respect to the corresponding observed outcome. This construction yields
an explicit closed-form expression in terms of the fitted logistic model
and population cell structure. We establish convergence of the effective
sample size ratio, asymptotic linearity of the resulting weighted
estimator, and a consistent plug-in variance estimator under a finite-cell
superpopulation framework. Simulation results examine population
recovery, regression performance, effective sample size, and the
finite-sample behavior of the asymptotic variance approximation. The
results show that logistic equivalent weighting provides competitive
finite-sample point estimation and approximately valid plug-in variance
estimation, effective sample size is competitive with alternative model-based weighting in the settings considered. An
application to the Future of Families and Child Wellbeing Study illustrates
the construction and use of the proposed weights and variance estimator
in an empirical survey setting. We emphasize the relationship between
logistic equivalent weighting and recent work on locally equivalent
weights for nonlinear multilevel regression and poststratification, while
focusing specifically on the frequentist interpretation of the weights,
their survey-weight properties, and their use for inference.
\end{abstract}

\doublespacing


\section{Introduction}

Survey weights are central to making sample data representative of target
populations. Constructing base weights or adjusting them involves multiple
steps and choices, often relying on auxiliary information to align
estimates with census counts or other population benchmarks. Survey data
also frequently face missingness and nonresponse, which can complicate
these procedures and contribute to unstable or highly variable weights.
These challenges motivate alternative approaches that use available
auxiliary information more directly in constructing population estimates.

A substantial body of literature addresses nonresponse adjustment, design- and model-based inference, and small area estimation (Chapman et al., \citeyear{Chapman1986};
 Little, \citeyear{Little1986};
 Chu and Goldman, \citeyear{ChuGoldman1997};
 Lu and Gelman, \citeyear{LuGelman2003};
 Rao and Molina, \citeyear{RaoMolina2015};
 Haziza and Beaumont, \citeyear{HazizaBeaumont2017};
 Skinner and Wakefield, \citeyear{SkinnerWakefield2017};
 Chen et al., \citeyear{ChenEtAl2017};
 Liu and Fan, \citeyear{LiuFan2023}). At the same time, concerns remain regarding the construction of base weights, particularly when detailed population information is unavailable or confidential. The role of weights in regression is also an active methodological question: weighting does not automatically reduce bias, and calibration procedures such as raking can produce highly variable weights (Little and Rubin, \citeyear{LittleRubin1987};
 Deville and Särndal, \citeyear{DevilleSarndal1992};
 Bethlehem et al., \citeyear{Bethlehem2011}).

Traditional approaches to constructing base weights follow design-based
principles (Chu and Goldman, \citeyear{ChuGoldman1997};
 Valliant et al., \citeyear{Valliant2013};
 Valliant and Dever, \citeyear{ValliantDever2018}). Model-based alternatives include regression-based
equivalent-weight constructions, most notably the framework discussed
by Gelman (\citeyear{Gelman2007}). For linear models, the poststratified model-based
population estimator can be written exactly as a weighted average of the
observed outcomes. For nonlinear models, however, the fitted model
coefficients depend on the observed outcomes, so an exact global
representation as a weighted sample mean with fixed weights is generally
unavailable. A local first-order representation can instead characterize
the sensitivity of the model-based population estimator to the observed
outcomes and thereby define locally equivalent weights.

Recent work by Giordano et al.~(\citeyear{GiordanoEtAl2026}) develops locally equivalent weights
for multilevel regression and poststratification (MrP), including
nonlinear outcome models such as logistic regression. Their framework
derives an asymptotic closed-form expression for logistic locally
equivalent weights from a Bayesian MrP formulation using
Bernstein--von Mises and delta-method approximations. Our derivation
proceeds independently from a frequentist perspective, starting from
the logistic estimating equations. We implicitly differentiate the
logistic score equation with respect to the observed outcomes to obtain
the sensitivity of the fitted coefficient vector, and then apply the
chain rule to the model-based population estimator. Under categorical
poststratification, this leads to the corresponding closed-form
weighting expression. Our focus, however, is on the interpretation of
these weights as local sensitivity measures and on their use as survey
weights for frequentist inference. We therefore study their effective
sample size, asymptotic linearity, plug-in variance estimation, and
finite-sample behavior.

The goal of this paper is to develop and evaluate logistic regression
equivalent weights as a bridge between nonlinear model-based population
estimation and survey-weight inference. We first define the equivalent
weights as the scaled first derivatives of the logistic population
estimator with respect to the observed outcomes and derive their
closed-form expression under categorical poststratification. This local
representation makes it possible to study the nonlinear estimator using
survey-weight concepts such as effective sample size and weight
concentration, while also providing a direct route to frequentist
inference through asymptotic linearization. We establish convergence of
the ESS ratio to a population limit and derive an asymptotic linear
representation with a consistent plug-in variance estimator under the
stated regularity conditions.

We then compare design-based, linear, and logistic weighting through
simulation and an empirical application. The simulation results show
that the logistic estimator provides competitive finite-sample
population-mean estimation, although it does not uniformly outperform
the linear or design-based alternatives. The results also demonstrate
favorable effective-sample-size behavior and approximately nominal
finite-sample performance of the proposed plug-in variance estimator.
The empirical application illustrates how the logistic equivalent
weights and associated variance estimator can be used for survey
inference with a single dataset.

The paper is organized as follows. Section~1 describes the target study
population, sampling design, and design-based weighting approach.
Section~2 develops the model-based weighting methods, including the
linear and logistic equivalent-weight constructions and the local
interpretation of the logistic weights. Section~3 establishes the
asymptotic framework and theoretical results. Section~4 presents the
simulation study and the application to the Fragile Families and Child
Wellbeing Study. Section~5 discusses the key findings, limitations,
relationship to recent work, and directions for future research.
Appendix~A provides a supplementary one-unit perturbation analysis,
while Appendix~B contains proofs of the theoretical results.

\subsection{Target Study Population and Design}

The target survey population consists of live births occurring in large U.S.\ cities with populations over 200,000 between 1998 and 2000. This focus is motivated by the process of constructing base weights for the Future of Families and Child Wellbeing Study (FFCWS). The FFCWS sample follows a stratified multistage design with 4,898 children, oversampling births to unmarried mothers at a ratio of 3 to 1, with the inclusion of a large number of Black, Hispanic, and low-income families (Reichman et al., \citeyear{ReichmanEtAl2001}). Follow-up interviews were conducted across seven waves, when children were approximately ages 1, 3, 5, 9, 15, and 22. In constructing the FFCWS weights, four demographic variables were used for poststratification, and geographic information was incorporated to estimate population birth counts using the Centers for Disease Control and Prevention (CDC) annual natality data. These variables are mother’s marital status, race/ethnicity, age, education, and city of birth (Carlson, \citeyear{Carlson2008}). To mirror the FFCWS study, we generate data and construct base weights in our simulation analysis that replicate the characteristics of the FFCWS design.

In large national surveys, stratified multistage cluster sampling (SMCS) is a common approach because it balances logistical challenges with the need to obtain sufficiently precise estimates for key subgroups. In this context, clusters, also referred to as primary sampling units (PSUs), are often used to simplify fieldwork. For example, the stratum is the city and the PSU is the hospital for cities selected with certainty in the FFCWS study. A simpler alternative is stratified simple random sampling (SRS), in which each unit in the population is assigned a group label and SRS is conducted within each group. These groups typically correspond to demographic strata or geographic regions. 

In our simulation study, we choose to construct design-based weights using stratified simple random sampling (SRS) rather than stratified multistage cluster sampling (SMCS). While SMCS is widely used, intra-cluster correlation reduces efficiency by limiting the independent information gained from additional sampled units within the same cluster. This leads to higher standard errors and lower overall precision. Moreover, we expect minimal differences in weight estimation between SMCS and SRS once raking is applied. For our experiments, we draw samples of size $n = 3000$ from a finite population of size $N = 1{,}000{,}000$, chosen to mirror realistic survey conditions where the sample represents only a small fraction of the population.

\subsection{Design-based Weighting}

For a finite population and a probability sample, each unit $i$ has a known selection probability $\pi_i > 0$, and the \emph{base weight} is defined as
\[
w_i = \frac{1}{\pi_i}.
\]
A probability sample is realized under the following conditions: the set of all possible samples that can be selected from the finite population is well-defined under the sampling design, and each possible sample is associated with a known probability of selection. Moreover, every unit in the target population has a nonzero probability of being selected, with selection occurring via a random mechanism (Valliant and Dever, 2018). In our simulation study, the selection probability $\pi_i$ is determined based on factors such as city of birth and demographic characteristics. These factors are used to form cells, where units within the same cell share identical probabilities of selection.

We then apply two standard adjustments: (i) nonresponse adjustment, where weights are multiplied by the inverse of the weighted response rate within adjustment cells (e.g., within each city) to account for survey nonrespondents; and (ii) raking, where weights are adjusted so that weighted sample counts align with known population margins. Nonresponse adjustments renormalize the weights to sum to the total population size, while raking aligns marginal totals without requiring full cross-classification, thereby avoiding the problem of empty cells. Given the extensive literature on nonresponse adjustments, we do not explore them in detail here. For recent applications to the FFCWS data, including a state-of-the-art two-stage approach using an optimally balanced Gaussian process, see (Vegetabile et al., \citeyear{Vegetabile2020}), refer to Lee and Gelman (\citeyear{LeeGelman2024}).
Let $X$ denote the poststratification variables whose joint distribution in the population is known, together with the survey response of interest $y$ that we aim to estimate for the population. The possible strata of $X$ form poststratification cells, denoted by $s$, with population sizes $N_s$ and sample sizes $n_s$. The total population size is $N=\sum_{s=1}^S N_s$ and the total sample size is $n=\sum_{s=1}^S n_s$.

The population mean of $y$ is
\[
\theta = \frac{\sum_{s=1}^S N_s\,\theta_s}{N},
\]
where $\theta_s$ is the population mean of $y$ within cell $s$. The corresponding sample-based poststratified estimator is
\[
\hat{\theta}^{PS}
=
\frac{\sum_{s=1}^S N_s\,\hat{\theta}_s}{N}.
\]

For a probability sample, the Horvitz--Thompson (HT) estimator is
\[
\bar y_{HT}=\frac{\sum_{i=1}^n w_i y_i}{N},
\]
where $w_i=1/\pi_i$ is the design weight. The Hájek estimator is
\[
\bar y_H=\frac{\sum_{i=1}^n w_i y_i}{\hat N},
\qquad
\hat N=\sum_{i=1}^n w_i.
\]
When the weights are normalized so that $\hat N=N$, the two expressions coincide algebraically. We use this normalization when comparing weighted estimators on the population-mean scale. Exact unbiasedness is a property of the Horvitz--Thompson estimator under the probability-sampling design; the Hájek estimator is generally a ratio estimator and need not be exactly unbiased in finite samples.

The poststratified estimator can be written as a weighted estimator when the cell-specific estimates $\hat\theta_s$ are represented through a regression or other model. We use this relationship to construct equivalent unit weights for the model-based methods below.

\section{Model-based Weighting Methods}

\subsection{Linear Regression Weight}

Let $X$ denote the $n \times k$ design matrix of auxiliary variables for the samples, 
and let $X^{pop}$ denote the $S \times k$ design matrix for the population poststratification 
cells. Here, $n$ is the sample size, $S$ is the number of poststratification cells, and 
$k$ is the number of auxiliary variables (e.g., demographic or geographic characteristics) 
used for weighting. If the survey response of interest $y$ has a linear relationship with the raking variables, then
\[
\hat{\theta}_s = X^{pop}_s \hat{\beta}, 
\qquad \hat{\beta} = (X^\top X)^{-1}X^\top y,
\]
where $\hat{\beta}$ are estimated regression coefficients. The poststratified estimate of the population mean is
\[
\hat{\theta}^{PS} 
= \frac{1}{N}\sum_{s=1}^S N_s \big( X^{pop}_s \hat{\beta} \big) 
= \frac{1}{N}(N^{pop})^\top X^{pop}(X^\top X)^{-1}X^\top y.
\]

Define $\hat{w}$ as the equivalent unit weights such that
\[
\hat{\theta}^{PS} = \frac{1}{n}\sum_{i=1}^n \hat{w}_i y_i, 
\qquad \sum_{i=1}^n \hat{w}_i = n.
\]
Solving for $\hat{w}$ gives
\[
\hat{w} = \frac{n}{N}(N^{pop})^\top X^{pop}(X^\top X)^{-1}X^\top.
\]

The outcome $y$ cancels out, so weights depend only on auxiliary variables. This is an advantage, as demographic and geographic information is typically available with minimal missingness. The formula produces $S$ unique equivalent unit weights, one for each poststratification cell. For linear unit weights, renormalization to sum to $n$ is unnecessary unless negative weights occur. Moreoever, multiplying the weights by any constant does not affect the estimation of the population mean. Thus, the weights can be renormalized to sum to the known population size, ensuring that 
$\sum_{i=1}^n \hat{w}_i = \hat{N}$. After this renormalization, the population estimate is 
\[
\hat{\theta}^{PS} = \frac{\sum_{i=1}^n \hat{w}_i y_i}{\hat{N}},
\]
,which corresponds to a form of the H\'ajek estimator.

\subsection{Logistic Regression Equivalent Weights}
\label{sec:logistic-equivalent-weights}

Let $y=(y_1,\ldots,y_n)^\top$ be a binary response vector and let
$X$ be the $n\times k$ matrix of weighting variables. We consider the
logistic regression model
\[
\Pr(y_i=1\mid X_i)=p_i,
\qquad
\operatorname{logit}(p_i)=X_i^\top\beta,
\]
with fitted probabilities
\[
\hat p_i=\sigma(X_i^\top\hat\beta),
\qquad
\sigma(t)=\frac{1}{1+e^{-t}}.
\]
For the population poststratification cells
$X_s^{\mathrm{pop}}$, $s=1,\ldots,S$, define the model-based population
estimator
\begin{equation}
\hat\theta^{LR}(y)
=
\frac{1}{N}
\sum_{s=1}^S
N_s
\sigma\!\left(
X_s^{\mathrm{pop}\top}\hat\beta(y)
\right),
\qquad
N=\sum_{s=1}^S N_s.
\label{eq:logit-pop-estimator}
\end{equation}
Here, $\hat\beta(y)$ emphasizes that the fitted logistic coefficient
vector is a function of the observed response vector.

For a linear regression model, the corresponding model-based population
estimator is linear in the observed outcomes and can therefore be written
exactly as a weighted sample mean. The logistic population estimator in
\eqref{eq:logit-pop-estimator}, however, is nonlinear in $y$ because
the fitted coefficient vector $\hat\beta(y)$ depends on the outcomes.
Consequently, we do not seek an exact global representation of the form
\[
\hat\theta^{LR}(y)
=
\frac{1}{n}\sum_{i=1}^n w_i y_i
\]
with weights that remain fixed as $y$ varies. Instead, we construct
weights that make a weighted sample mean locally equivalent to the
logistic population estimator at the observed data.

To formalize this local equivalence, let $y_0$ denote the observed
response vector and consider the path
\[
y(t)=y_0+t h,
\]
where $h\in\mathbb{R}^n$ is an arbitrary direction. The scalar $t$ is
only a mathematical device for defining a derivative; no actual
perturbation of the observed binary outcomes is required. We evaluate
the resulting derivative at $t=0$, corresponding to the observed data
$y_0$.

First, a first-order Taylor expansion of the logistic population
estimator around the observed data gives
\begin{equation}
\hat\theta^{LR}(y_0+t h)
=
\hat\theta^{LR}(y_0)
+
t
\nabla_y\hat\theta^{LR}(y_0)^\top h
+
o(t),
\qquad t\to0.
\label{eq:logit-taylor}
\end{equation}
Thus, the first-order change in the logistic population estimator is
\begin{equation}
\Delta\hat\theta^{LR}
=
t
\nabla_y\hat\theta^{LR}(y_0)^\top h
+
o(t).
\label{eq:logit-taylor-change}
\end{equation}

Next, consider an ordinary weighted sample mean with a weight vector
$w$ evaluated at the observed data $y_0$ and held fixed along the local
perturbation:
\[
T_w(y)
=
\frac{1}{n}w^\top y.
\]
Because $T_w(y)$ is linear in $y$, its change along the same path is
exactly
\begin{equation}
T_w(y_0+t h)-T_w(y_0)
=
\frac{t}{n}w^\top h.
\label{eq:weighted-taylor}
\end{equation}

We define the local equivalent weights by requiring these two first-order
changes to agree for every direction $h$:
\begin{equation}
\nabla_y\hat\theta^{LR}(y_0)^\top h
=
\frac{1}{n}w^\top h
\qquad
\text{for all }h.
\label{eq:local-equivalence-condition}
\end{equation}
Because the equality in \eqref{eq:local-equivalence-condition} holds
for every direction $h$, the corresponding coefficient vectors must be
equal. Hence, the local equivalent weights are defined by
\begin{equation}
\boxed{
\hat w
=
n
\left(
\frac{\partial\hat\theta^{LR}}
{\partial y^\top}
\right)^\top
\bigg|_{y=y_0}.
}
\label{eq:logit-local-definition}
\end{equation}
Equivalently, for observation $i$,
\begin{equation}
\boxed{
\hat w_i
=
n
\frac{\partial\hat\theta^{LR}}
{\partial y_i}
\bigg|_{y=y_0}.
}
\label{eq:logit-individual-weight}
\end{equation}

Thus, the proposed weights are not obtained by assuming that the weights
are globally independent of the outcomes. Rather, they are defined as
the coefficients that make a fixed-weight linear estimator have the same
first-order sensitivity as the nonlinear logistic population estimator
at the observed data.

We next derive the gradient in \eqref{eq:logit-local-definition}.

Define the sample and population logistic variance matrices by
\[
W(\beta)
=
\operatorname{diag}
\left\{
\sigma(X\beta)\odot[1-\sigma(X\beta)]
\right\}
\]
and
\[
W^{\mathrm{pop}}(\beta)
=
\operatorname{diag}
\left\{
\sigma(X^{\mathrm{pop}}\beta)
\odot
[1-\sigma(X^{\mathrm{pop}}\beta)]
\right\}.
\]
Let
\[
N^{\mathrm{pop}}
=
(N_1,\ldots,N_S)^\top.
\]

The logistic regression estimator satisfies the score equation
\[
X^\top(y-\hat p)=0,
\qquad
\hat p=\sigma(X\hat\beta).
\]
Because $\hat\beta$ is implicitly determined by this equation, implicit
differentiation with respect to $y$ gives
\begin{equation}
\frac{\partial\hat\beta}
{\partial y^\top}
=
(X^\top W(\hat\beta)X)^{-1}X^\top.
\label{eq:dbetady-main}
\end{equation}

On the other hand, differentiating the model-based population estimator
in \eqref{eq:logit-pop-estimator} with respect to $\hat\beta$ gives
\begin{equation}
\frac{\partial\hat\theta^{LR}}
{\partial\hat\beta^\top}
=
\frac{1}{N}
(N^{\mathrm{pop}})^\top
W^{\mathrm{pop}}(\hat\beta)
X^{\mathrm{pop}}.
\label{eq:dthetadb-main}
\end{equation}

Applying the chain rule,
\begin{equation}
\frac{\partial\hat\theta^{LR}}
{\partial y^\top}
=
\frac{\partial\hat\theta^{LR}}
{\partial\hat\beta^\top}
\frac{\partial\hat\beta}
{\partial y^\top}.
\end{equation}
Substituting \eqref{eq:dbetady-main} and
\eqref{eq:dthetadb-main} yields
\begin{equation}
\frac{\partial\hat\theta^{LR}}
{\partial y^\top}
=
\frac{1}{N}
(N^{\mathrm{pop}})^\top
W^{\mathrm{pop}}(\hat\beta)
X^{\mathrm{pop}}
(X^\top W(\hat\beta)X)^{-1}
X^\top.
\label{eq:logit-gradient}
\end{equation}

Combining \eqref{eq:logit-gradient} with the definition of the local
equivalent weights in \eqref{eq:logit-local-definition} gives the
closed-form expression
\begin{equation}
\boxed{
\hat w(\hat\beta)
=
\frac{n}{N}
X
(X^\top W(\hat\beta)X)^{-1}
(X^{\mathrm{pop}})^\top
W^{\mathrm{pop}}(\hat\beta)
N^{\mathrm{pop}}.
}
\label{eq:logit-weights}
\end{equation}

The resulting weights therefore satisfy the local first-order
representation
\begin{equation}
\hat\theta^{LR}(y_0+t h)
=
\hat\theta^{LR}(y_0)
+
\frac{t}{n}\hat w^\top h
+
o(t),
\qquad
t\to0.
\label{eq:logit-local-expansion}
\end{equation}
Equivalently, for a general infinitesimal perturbation $\Delta y$,
\begin{equation}
\hat\theta^{LR}(y_0+\Delta y)
=
\hat\theta^{LR}(y_0)
+
\frac{1}{n}\hat w^\top\Delta y
+
o(\|\Delta y\|).
\label{eq:logit-local-expansion-delta}
\end{equation}
Thus, $\hat w_i/n$ is the local sensitivity of the model-based
population estimator to observation $i$'s outcome, evaluated at the
observed data. The vector $\hat w$ provides a local equivalent-weight
representation of the nonlinear logistic population estimator, rather
than an exact global representation as a weighted sample mean.

Because the population is represented by $S$ poststratification cells,
the resulting weights depend on the weighting variables and the fitted
logistic model. In particular, sampled units with identical
weighting-variable patterns receive the same local equivalent weight,
so there are at most $S$ distinct raw equivalent weights.

For estimation of a population mean, the weights are invariant to a
common multiplicative rescaling when used in a normalized weighted mean.
We therefore normalize the raw equivalent weights to sum to $n$ when
reporting equivalent unit weights. This normalization does not alter the
corresponding normalized weighted mean or the effective sample size.

\subsection{Interpretation of Local Equivalent Weights}
\label{sec:interpretation-weights}

The local equivalent weights have a direct interpretation in terms of the
sensitivity of the model-based population estimator to perturbations in the
observed outcomes. By construction,
\begin{equation}
\hat w_i
=
n
\frac{\partial \hat\theta_{LR}}
{\partial y_i},
\label{eq:weight-sensitivity}
\end{equation}
where the derivative is evaluated at the observed outcome vector. Thus, for
a small perturbation $\Delta y_i$ to observation $i$, holding the remaining
outcomes fixed,
\begin{equation}
\Delta\hat\theta_{LR}
\approx
\frac{\hat w_i}{n}\Delta y_i.
\label{eq:individual-sensitivity}
\end{equation}
More generally, for a perturbation $\Delta y$ to the entire outcome vector,
the first-order expansion is
\begin{equation}
\hat\theta_{LR}(y+\Delta y)
=
\hat\theta_{LR}(y)
+
\frac{1}{n}\hat w^{\top}\Delta y
+
o(\|\Delta y\|).
\label{eq:local-equivalence}
\end{equation}
Thus, $\hat w_i$ measures the local sensitivity of the model-based
population estimator to observation $i$'s outcome, while the sign determines
the direction of that sensitivity.

A positive equivalent weight indicates that an increase in $y_i$ locally
increases the estimated population quantity. Conversely, a negative
equivalent weight indicates that an increase in $y_i$ locally decreases the
estimated population quantity. For a binary outcome, a positive weight
therefore implies that changing $y_i$ from $0$ to $1$ locally increases the
estimated population prevalence, whereas a negative weight implies that the
same change locally decreases the estimated population prevalence.

This interpretation differs fundamentally from that of conventional design
or inverse-probability weights. A negative equivalent weight does not
represent a negative sampling probability or a negative number of population
units represented by an observation. Rather, it reflects the direction of
the observation's local contribution to the model-based population prediction
through the fitted regression model.

The possibility of negative weights follows directly from the closed-form
expression in \eqref{eq:logit-weights}. For observation $i$,
\begin{equation}
\hat w_i
=
\frac{n}{N}
x_i^{\top}
(X^{\top}WX)^{-1}
X_{\mathrm{pop}}^{\top}
W_{\mathrm{pop}}
N_{\mathrm{pop}}.
\label{eq:individual-weight-form}
\end{equation}
The first factor describes how a perturbation to observation $i$ affects the
fitted logistic coefficients, while the second describes how changes in
those coefficients affect the population-level predicted outcome. These two
effects need not have the same direction. Consequently, their combined
effect can be positive or negative.

Negative values are therefore not, by themselves, a mathematical failure of
the local equivalent-weight construction. They are a natural consequence of
using regression sensitivity to represent a nonlinear model-based estimator.
However, negative or highly variable weights can create practical
difficulties when the resulting weights are subsequently used as
conventional survey weights, particularly for procedures such as raking that
require strictly positive starting weights.

The magnitude of the weights also has a direct interpretation. Large
absolute values of $\hat w_i$ indicate that the model-based population
estimator is highly sensitive, locally, to the outcome of observation $i$.
Accordingly, the distribution of the raw local equivalent weights provides
information about the stability and concentration of the model-based
population estimator. We therefore examine both the effective sample size
and the occurrence of negative or highly variable weights in the simulation
study.

The implications of this interpretation are illustrated by the empirical
distribution of the weights in the simulation study. Figure~\ref{fig:weight-distributions}
shows the distributions of the design-based and model-based weights across
the 300 simulated samples. The raw linear and logistic equivalent weights
exhibit substantially different distributions from the design-based weights,
including greater concentration near zero, heavier right tails, and
occasional negative values. These features reflect the fact that the
model-based weights are local sensitivity coefficients rather than
inverse-probability weights.

\begin{figure}[htbp]
    \centering
    \includegraphics[width=\textwidth]{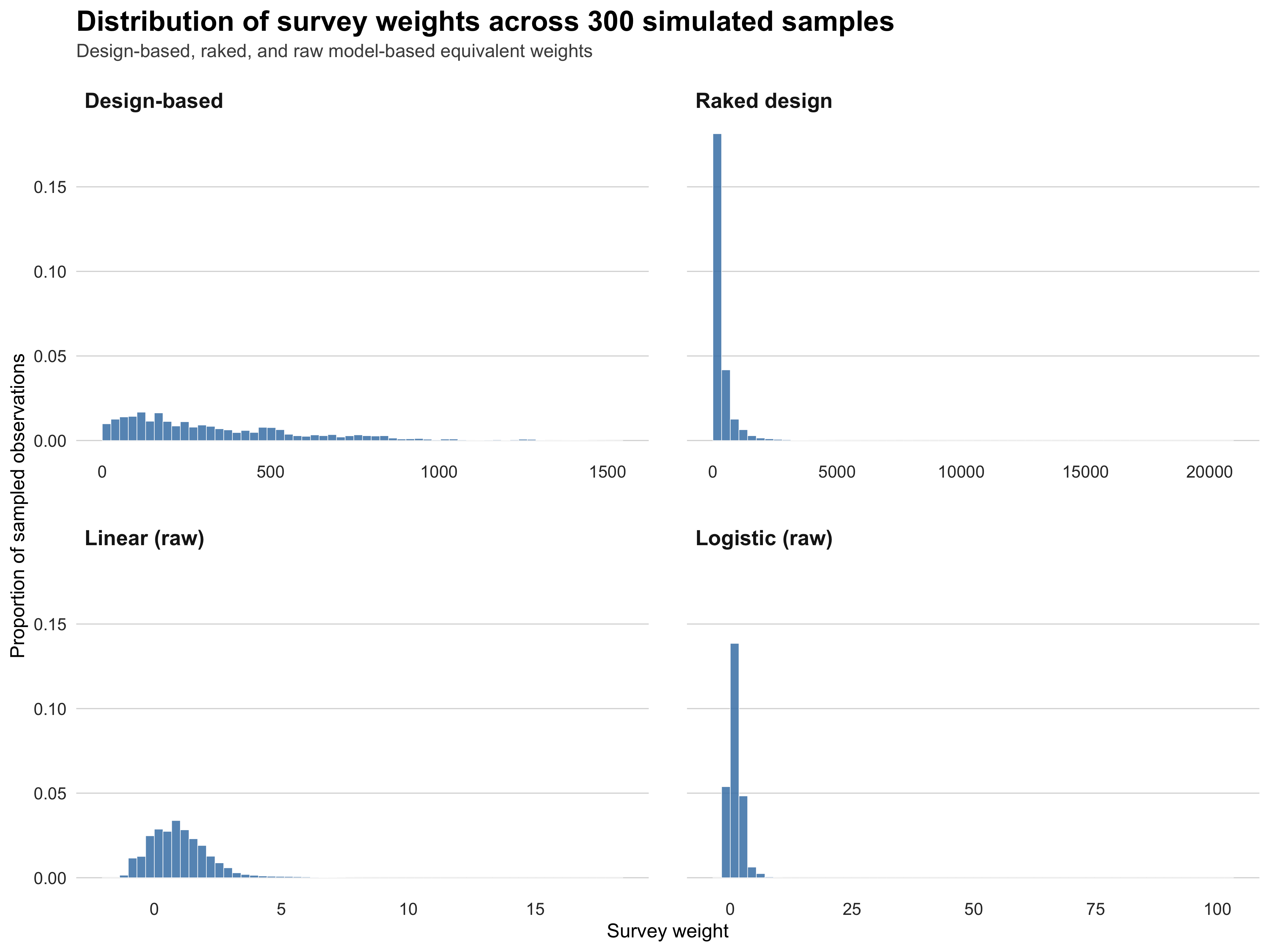}
    \caption{Distribution of design-based and model-based equivalent
    weights across 300 simulated samples. The raw linear and logistic
    weights illustrate the variability and occasional negative values
    that can arise from model-based equivalent-weight constructions.
    Raked weights are shown for comparison after the positivity
    adjustment and calibration procedure.}
    \label{fig:weight-distributions}
\end{figure}

\subsection{Practical Treatment of Nonpositive Weights}
\label{sec:practical-weights}

The theoretical development above concerns the raw local equivalent weights
defined in \eqref{eq:logit-weights}. In the empirical comparisons reported
below, we subsequently use raking to impose the known population margins.
Because the raking procedure requires strictly positive starting weights,
nonpositive model-based weights must first be modified.

For the practical comparisons, each nonpositive model-based weight is
replaced by the smallest strictly positive weight observed among the sample
observations for the corresponding weighting method. Specifically, let
$\hat w_i$ denote the raw model-based equivalent weight and define
\begin{equation}
w_i^{*}
=
\begin{cases}
\hat w_i, &
\hat w_i > 0,\\[4pt]
\displaystyle
\min_{j:\,\hat w_j>0}\hat w_j, &
\hat w_i \leq 0.
\end{cases}
\label{eq:positive-adjustment}
\end{equation}
The modified weights are then renormalized to sum to the sample size:
\begin{equation}
\tilde w_i
=
\frac{n w_i^{*}}
{\displaystyle\sum_{j=1}^{n}w_j^{*}}.
\label{eq:renormalized-weights}
\end{equation}
The resulting positive weights are used as the starting weights for the
raking procedure.

This modification is a practical post-processing step rather than part of
the proposed local equivalent-weight construction. In particular, the
asymptotic results in Section~3 concern the raw equivalent weights
$\hat w_i$ in \eqref{eq:logit-weights}, not the modified weights
$\tilde w_i$. Consequently, the theoretical results do not establish
asymptotic properties for the positivity-adjusted or subsequently raked
estimators.

The positivity adjustment is applied consistently across the model-based
weighting methods in the empirical comparisons. This allows the linear and
logistic constructions to be evaluated under the same practical requirement
that the starting weights be admissible for raking. The adjustment should
therefore be viewed as a stabilization and implementation step rather than
as part of the definition of the proposed estimator.

Alternative constrained transformations could be considered to enforce
positivity while preserving additional calibration constraints. Such
procedures are outside the scope of the present study. We instead report the
frequency and magnitude of nonpositive weights separately and evaluate the
practical consequences of the positivity adjustment empirically.

\subsection{Missingness and Nonresponse}

In practice, missing data in demographic and geographic variables used for survey weighting should be assessed before constructing weights. When auxiliary variables are missing, both design-based and model-based weighting become more difficult because the corresponding population and sample cells may no longer be directly comparable.

For nonresponse adjustment, design-based methods typically require an explicit adjustment step based on response rates within adjustment cells. The model-based construction studied here instead uses auxiliary variables directly in the regression model and population poststratification. This does not eliminate nonresponse assumptions: the validity of the model-based estimator still depends on whether the selected weighting variables adequately capture systematic differences relevant to the target population.

In the FFCWS application, city-level population information is not directly available in the CDC natality data and must be constructed using sample information and external population totals. We therefore treat the population cell counts used in the application as fixed inputs to the weighting procedures.

\section{Asymptotic Properties}

\subsection{Finite-population and superpopulation frameworks}

The weighting construction is motivated by a finite-population survey
problem in which the population cell sizes $N_s$ are treated as known
or estimated external benchmarks. The asymptotic results in this
section are instead established under a finite-cell superpopulation
framework, which provides a tractable large-sample approximation for
the estimated equivalent weights.

Specifically, we consider a finite set of poststratification cells and
assume that the sampled covariates $X_i$ are generated from a
distribution that assigns positive probability to the cells relevant
to the target population. The asymptotic analysis therefore studies
repeated samples from this superpopulation while treating the
target-population cell proportions as fixed. Accordingly, the results
below should be interpreted as superpopulation results under the stated
regularity conditions, rather than as a general design-based
asymptotic theory for complex surveys.

The distinction is important. The proposed weights are motivated by
survey calibration and poststratification, whereas the results below
provide model-based large-sample results for the limiting
equivalent-weight construction. Extensions to complex sampling
designs, including clustering and unequal-probability multistage
sampling, are left for future work.

\subsection{Population-limit equivalent weights}

Let the finite set of poststratification cells be
$s=1,\ldots,S$. Define the sample-cell probabilities
\[
\pi_s=P(X=x_s)>0,
\qquad s=1,\ldots,S,
\]
and the fixed target-population cell proportions
\[
q_s=\frac{N_s}{N},
\qquad
\sum_{s=1}^S q_s=1.
\]

For $\beta$ in a neighborhood of $\beta_0$, define
\[
p_s(\beta)
=
\sigma(x_s^\top\beta),
\qquad
v_s(\beta)
=
p_s(\beta)\{1-p_s(\beta)\}.
\]

For a generic vector of sample-cell probabilities
$\pi=(\pi_1,\ldots,\pi_S)$, define the
\emph{population logistic information matrix}
\[
Q(\beta,\pi)
=
\sum_{s=1}^S
\pi_s v_s(\beta)x_sx_s^\top
\]
and the corresponding target-population score vector
\[
C(\beta)
=
\sum_{s=1}^S
q_s v_s(\beta)x_s.
\]

Assume that $Q(\beta,\pi)$ is nonsingular on a neighborhood of
interest. The population-limit raw equivalent weight for an observation
in cell $s$ is defined by
\begin{equation}
w_s^{\mathrm{raw}}(\beta,\pi)
=
x_s^\top Q(\beta,\pi)^{-1}C(\beta).
\label{eq:population-raw-weight}
\end{equation}

Equivalently,
\[
w^{\mathrm{raw}}(X;\beta,\pi)
=
w_s^{\mathrm{raw}}(\beta,\pi)
\qquad
\text{when }X=x_s.
\]

For the observed sample, let
\[
\hat\pi_s
=
\frac{n_s}{n}
=
\frac{1}{n}\sum_{i=1}^n I(X_i=x_s).
\]
The corresponding finite-sample raw equivalent weight is
\begin{equation}
\hat w_i^{\mathrm{raw}}(\beta)
=
w^{\mathrm{raw}}(X_i;\beta,\hat\pi).
\label{eq:finite-raw-weight}
\end{equation}

Because the equivalent weights are identified only up to a common
multiplicative constant, the population target associated with the raw
weights is naturally written as the normalized functional
\[
T(\beta,\pi)
=
\frac{
E[w^{\mathrm{raw}}(X;\beta,\pi)Y]
}{
E[w^{\mathrm{raw}}(X;\beta,\pi)]
}.
\]

We define the population-limit equivalent-weight target by
\[
\theta_{\mathrm{EW}}=T(\beta_0,\pi).
\]

The corresponding sample estimator is
\begin{equation}
\hat\theta_{\mathrm{EW}}
=
\frac{
P_n[
\hat w^{\mathrm{raw}}(\hat\beta)Y
]
}{
P_n[
\hat w^{\mathrm{raw}}(\hat\beta)
]
},
\qquad
P_nf=\frac1n\sum_{i=1}^n f(X_i,Y_i).
\label{eq:normalized-ew-estimator}
\end{equation}

For estimation of the population mean, the raw equivalent weights can be
rescaled by a common nonzero constant; in particular, when their sum is
nonzero, they may be normalized to sum to $n$. This normalization changes
only the scale of the individual weights and does not change the resulting
weighted mean or effective sample size. The normalized representation is
therefore equivalent to rescaling the raw equivalent weights to sum to
$n$. Because the normalized weights depend on both $\hat{\beta}$ and
$\hat{\pi}$, this representation makes explicit that the sampling
variation arises both from the observed outcomes and from estimation of
the cell proportions entering the equivalent weights.
\subsection{Regularity conditions}

\begin{lemma}[]
\label{lem:weight-regularity}

Suppose that $X$ consists entirely of categorical variables encoded as
dummies, so that
\[
\mathcal X=\{x_1,\ldots,x_S\}
\]
is finite and
\[
\sup_{x\in\mathcal X}\|x\|<\infty.
\]

Let $\mathcal N$ be a compact neighborhood of $\beta_0$ and let
$\Pi$ be a neighborhood of the population cell-probability vector
$\pi$. Suppose there exists $c>0$ such that
\[
\lambda_{\min}\{Q(\beta,\tilde\pi)\}
\ge c
\]
for all $(\beta,\tilde\pi)\in\mathcal N\times\Pi$.

Then:

\begin{enumerate}

\item[(i)] The map
\[
(\beta,\tilde\pi)
\mapsto
w^{\mathrm{raw}}(x;\beta,\tilde\pi)
\]
is continuous on $\mathcal N\times\Pi$ for every
$x\in\mathcal X$.

\item[(ii)] There exists a finite constant $C$ such that
\[
\sup_{(\beta,\tilde\pi)\in\mathcal N\times\Pi}
\max_{x\in\mathcal X}
\left|
w^{\mathrm{raw}}(x;\beta,\tilde\pi)
\right|
\le C.
\]

\end{enumerate}

Consequently, because $\hat\pi\xrightarrow{p}\pi$, the finite-sample
weights satisfy
\[
\sup_{\beta\in\mathcal N}
\max_{1\le i\le n}
\left|
\hat w_i^{\mathrm{raw}}(\beta)
\right|
=O_p(1).
\]

\end{lemma}

\subsection{Effective sample size}

\begin{theorem}[]
\label{thm:plugin-LLN}

Let $\{(X_i,Y_i)\}_{i=1}^n$ be i.i.d., with $X_i$ taking values in the
finite set $\mathcal X$. Suppose the conditions of
Lemma~\ref{lem:weight-regularity} hold and
\[
\hat\beta\xrightarrow{p}\beta_0.
\]

Define
\[
\bar w_n(\beta,\tilde\pi)
=
\frac1n
\sum_{i=1}^n
w^{\mathrm{raw}}(X_i;\beta,\tilde\pi)
\]
and
\[
\overline{w^2}_n(\beta,\tilde\pi)
=
\frac1n
\sum_{i=1}^n
\left\{
w^{\mathrm{raw}}(X_i;\beta,\tilde\pi)
\right\}^2.
\]

Then
\[
\bar w_n(\hat\beta,\hat\pi)
\xrightarrow{p}
E\left[
w^{\mathrm{raw}}(X;\beta_0,\pi)
\right],
\]
and
\[
\overline{w^2}_n(\hat\beta,\hat\pi)
\xrightarrow{p}
E\left[
w^{\mathrm{raw}}(X;\beta_0,\pi)^2
\right].
\]

Consequently, the effective sample size
\[
\mathrm{ESS}_n
=
\frac{
\left\{
\sum_{i=1}^n
\hat w_i^{\mathrm{raw}}(\hat\beta)
\right\}^2
}{
\sum_{i=1}^n
\left\{
\hat w_i^{\mathrm{raw}}(\hat\beta)
\right\}^2
}
\]
satisfies
\begin{equation}
\frac{\mathrm{ESS}_n}{n}
\xrightarrow{p}
\frac{
\left[
E\{w^{\mathrm{raw}}(X;\beta_0,\pi)\}
\right]^2
}{
E\left[
w^{\mathrm{raw}}(X;\beta_0,\pi)^2
\right]
}.
\label{eq:ESS-limit}
\end{equation}

\end{theorem}

\begin{remark}[Scale invariance]
If all equivalent weights are multiplied by a common positive constant,
both the normalized estimator and the effective sample size are
unchanged. The raw-weight representation is therefore used for
derivation, while normalized weights may be used for reporting.
\end{remark}

\subsection{Asymptotic linearity}

We next characterize the sampling variation of the normalized
equivalent-weight estimator. Two estimated quantities enter the
weights: the logistic regression coefficient $\hat\beta$ and the
empirical cell-probability vector $\hat\pi$. Their contributions are
therefore treated separately.

Define
\[
a_0
=
E\left[
w^{\mathrm{raw}}(X;\beta_0,\pi)
\right].
\]

The sensitivity of the population equivalent-weight estimator to the
logistic regression coefficient is
\begin{equation}
B_\beta
=
\left.
\frac{\partial T(\beta,\pi)}
{\partial\beta}
\right|_{(\beta_0,\pi)}.
\label{eq:B-beta}
\end{equation}

Similarly, define the sensitivity to the sample-cell probabilities by
\begin{equation}
B_\pi
=
\left.
\frac{\partial T(\beta_0,\tilde\pi)}
{\partial\tilde\pi}
\right|_{\tilde\pi=\pi}.
\label{eq:B-pi}
\end{equation}

The derivatives in \eqref{eq:B-beta} and \eqref{eq:B-pi} describe how
the final population estimator changes when the fitted logistic
coefficient or the empirical cell composition changes.

\begin{lemma}[]
\label{lem:weight-differentiability}

Under the assumptions of
Lemma~\ref{lem:weight-regularity}, suppose additionally that
$Q(\beta,\tilde\pi)$ is nonsingular throughout
$\mathcal N\times\Pi$.

Then
\[
(\beta,\tilde\pi)
\mapsto
w^{\mathrm{raw}}(x;\beta,\tilde\pi)
\]
is continuously differentiable for every $x\in\mathcal X$.
Moreover, there exist finite constants $C_\beta$ and $C_\pi$ such that
\[
\sup_{(\beta,\tilde\pi)\in\mathcal N\times\Pi}
\left\|
\nabla_\beta
w^{\mathrm{raw}}(x;\beta,\tilde\pi)
\right\|
\le C_\beta
\]
and
\[
\sup_{(\beta,\tilde\pi)\in\mathcal N\times\Pi}
\left\|
\nabla_\pi
w^{\mathrm{raw}}(x;\beta,\tilde\pi)
\right\|
\le C_\pi
\]
for every $x\in\mathcal X$.

Consequently, the functional $T(\beta,\pi)$ is continuously
differentiable in both arguments in a neighborhood of
$(\beta_0,\pi)$.

\end{lemma}

\subsection{Asymptotic linearity and plug-in variance}

\begin{theorem}[]
\label{thm:AL}

Suppose the conditions of
Lemmas~\ref{lem:weight-regularity} and
\ref{lem:weight-differentiability} hold. In addition, assume:

\begin{enumerate}

\item[(i)] The logistic regression model is correctly specified at
$\beta_0$, so that
\[
p_0(X)
=
\sigma(X^\top\beta_0).
\]

\item[(ii)] $\beta_0$ is an interior point of the parameter space and
the Fisher information matrix
\[
\mathcal I(\beta_0)
=
E\left[
XX^\top p_0(X)\{1-p_0(X)\}
\right]
\]
is positive definite.

\item[(iii)] The logistic maximum likelihood estimator satisfies
the standard asymptotic linear representation
\[
\sqrt n(\hat\beta-\beta_0)
=
\frac1{\sqrt n}
\sum_{i=1}^n
\mathcal I(\beta_0)^{-1}
X_i
\{Y_i-p_0(X_i)\}
+
o_p(1).
\]

\end{enumerate}

Let
\[
e(X)
=
\left(
I(X=x_1),\ldots,I(X=x_S)
\right)^\top.
\]

Then
\[
\sqrt n
\left(
\hat\theta_{\mathrm{EW}}
-
\theta_{\mathrm{EW}}
\right)
=
\frac1{\sqrt n}
\sum_{i=1}^n
\psi(Z_i)
+
o_p(1),
\]
where
\[
Z_i=(X_i,Y_i)
\]
and the influence function is
\begin{equation}
\boxed{
\psi(Z)
=
\frac{
w^{\mathrm{raw}}(X;\beta_0,\pi)
}{
a_0
}
\{Y-\theta_{\mathrm{EW}}\}
+
B_\beta^\top
\mathcal I(\beta_0)^{-1}
X
\{Y-p_0(X)\}
+
B_\pi^\top
\{e(X)-\pi\}.
}
\label{eq:EW-influence}
\end{equation}

Consequently,
\[
\sqrt n
\left(
\hat\theta_{\mathrm{EW}}
-
\theta_{\mathrm{EW}}
\right)
\Rightarrow
N(0,V),
\]
where
\[
V
=
E\left[\psi(Z)^2\right].
\]

\end{theorem}

\begin{remark}[Interpretation of the influence function]
The three terms in \eqref{eq:EW-influence} correspond to three distinct
sources of first-order sampling variation. The first term is the
ordinary variation in the weighted outcome mean when the limiting
weights are treated as fixed. The second term accounts for estimation
of the logistic regression coefficient $\beta_0$. The third term
accounts for estimation of the sample-cell probabilities $\pi$ used
to construct the equivalent weights.

Thus, the asymptotic variance incorporates both outcome variation and
the additional uncertainty induced by estimating the model and the
sample-cell composition.
\end{remark}

\subsection{Plug-in variance estimator}

For practical variance estimation, define
\[
\hat a
=
\frac1n
\sum_{i=1}^n
\hat w_i^{\mathrm{raw}}(\hat\beta),
\]
and let the raw equivalent weights be normalized by
\[
\tilde w_i
=
\frac{
\hat w_i^{\mathrm{raw}}(\hat\beta)
}{
\hat a
}.
\]

The empirical analogue of the first component of
\eqref{eq:EW-influence} is therefore
\[
\tilde w_i(Y_i-\hat\theta_{\mathrm{EW}}).
\]

Let
\[
\hat p_i
=
\sigma(X_i^\top\hat\beta),
\]
and define
\[
\hat{\mathcal I}
=
\frac1n
\sum_{i=1}^n
X_iX_i^\top
\hat p_i(1-\hat p_i).
\]

Define
\[
\hat B_\beta
=
\left.
\frac{\partial T(\beta,\hat\pi)}
{\partial\beta}
\right|_{\beta=\hat\beta}
\]
and
\[
\hat B_\pi
=
\left.
\frac{\partial T(\hat\beta,\tilde\pi)}
{\partial\tilde\pi}
\right|_{\tilde\pi=\hat\pi}.
\]

Then the estimated influence contribution for observation $i$ is
\begin{equation}
\hat\psi_i
=
\tilde w_i
(Y_i-\hat\theta_{\mathrm{EW}})
+
\hat B_\beta^\top
\hat{\mathcal I}^{-1}
X_i
(Y_i-\hat p_i)
+
\hat B_\pi^\top
\{e(X_i)-\hat\pi\}.
\label{eq:estimated-influence}
\end{equation}

A plug-in variance estimator is
\begin{equation}
\boxed{
\widehat V
=
\frac1n
\sum_{i=1}^n
\hat\psi_i^2.
}
\label{eq:plugin-variance}
\end{equation}

Under the regularity conditions of Theorem~\ref{thm:AL},
\[
\widehat V\xrightarrow{p}V.
\]

\section{Results}
\subsection{Simulation Results}

We first examine the performance of the weighting methods in regression
modeling and their effect on effective sample size. Figure~\ref{fig:MSE}
reports the pooled mean squared error (MSE) of the regression
coefficients across 300 independent resamples. In the simulated setting,
both model-based approaches---linear and logistic weighting---yield lower
average MSE than the design-based and unweighted approaches, although the
MSE varies substantially across resamples. Thus, the lower average MSE of
the model-based approaches does not necessarily imply uniformly greater
stability in the regression estimates.

Figure~\ref{fig:ESS} presents the effective sample size (ESS) of the
three weighting methods across the 300 simulated samples, before and
after raking. Before raking, the design-based weights have the lowest
ESS, while the linear and logistic equivalent weights have substantially
larger and similar ESS values. After raking, the pattern changes:
the design-based weights have the highest ESS, whereas the linear and
logistic weights have lower ESS, with the logistic weights yielding the
lowest ESS among the three methods. Thus, in this simulation, raking
increases the ESS of the design-based weights but reduces the ESS of
the model-based equivalent weights. These results illustrate that
raking can substantially alter the concentration of the weights and
that the effect of raking differs across weighting constructions.

\begin{figure}[H]
  \centering
  \begin{subfigure}{0.5\textwidth}
    \centering
    \includegraphics[width=\linewidth]{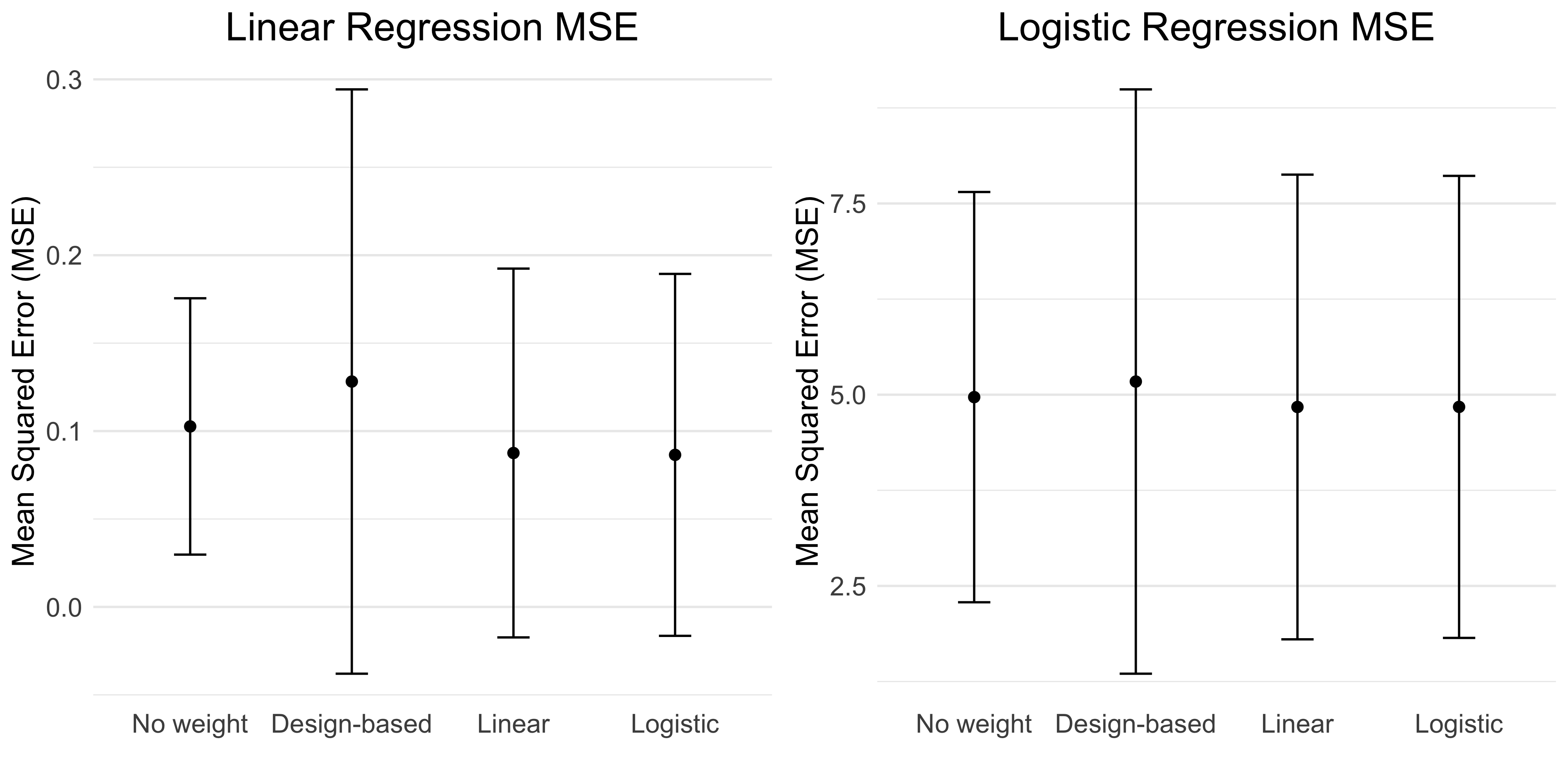}
    \caption{MSE across resamples.}
    \label{fig:MSE}
  \end{subfigure}\hfill
  \begin{subfigure}{0.5\textwidth}
    \centering
    \includegraphics[width=\linewidth]{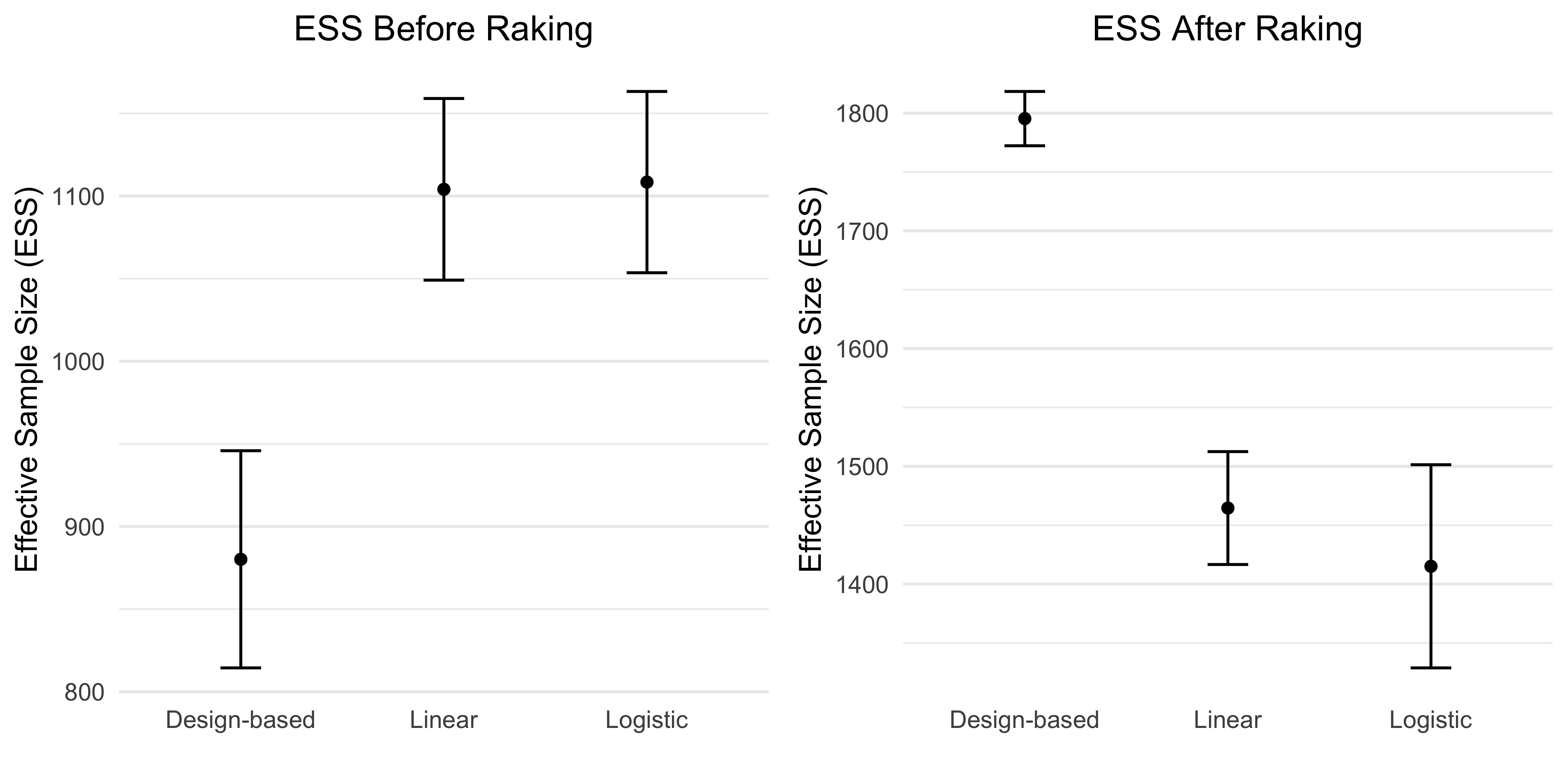}
    \caption{ESS across resamples ($n=3000$).}
    \label{fig:ESS}
  \end{subfigure}
\end{figure}

\FloatBarrier
Next, we evaluate the finite-sample performance of the three
equivalent-weight estimators in estimating the population mean of a
binary outcome. The simulation target is the superpopulation mean
\[
\theta_0
=
\frac{1}{N}\sum_{j=1}^N p_j,
\]
where $p_j$ denotes the true outcome probability used to generate the
binary responses. Table~\ref{tab:simulation_performance} reports the
Monte Carlo mean of the estimator, bias, mean squared error (MSE), and
empirical standard deviation.

\begin{table}[H]
\centering
\caption{Finite-sample performance of equivalent-weight estimators}
\label{tab:simulation_performance}
\begin{tabular}{lrrrr}
\hline
Method & Estimate & Bias & MSE & Empirical SD \\
\hline
Design-based & 0.352 & -0.021 & 0.000595 & 0.0121 \\
Linear       & 0.374 &  0.001 & 0.000148 & 0.0121 \\
Logistic     & 0.369 & -0.004 & 0.000194 & 0.0134 \\
\hline
Population mean & \multicolumn{4}{c}{0.373} \\
\hline
\end{tabular}
\end{table}

The linear estimator has the smallest MSE in this simulation and is
approximately unbiased, while the logistic estimator also exhibits
small finite-sample bias and competitive MSE. The design-based
estimator has the largest bias and MSE under the data-generating
process considered. These results illustrate that the proposed
logistic equivalent weights provide competitive finite-sample point
estimation, although they do not uniformly dominate the alternative
weighting methods.

Although the overall population mean provides the primary measure of
finite-sample performance, it can mask heterogeneity in estimation
across population subgroups. We therefore conduct a secondary
city-level analysis using the 16 city categories included in the
simulation. For city $c$, let
\[
\theta_c
=
\frac{1}{N_c}
\sum_{j:C_j=c} p_j,
\]
denote the corresponding population prevalence, where $N_c$ is the
population size of city $c$. Figure~\ref{fig:city_prevalence} shows
the distribution of the estimated city-level prevalence across the
300 simulated samples, with the corresponding population prevalence
indicated by the red points.

\begin{figure}[H]
    \centering
    \includegraphics[width=\textwidth]{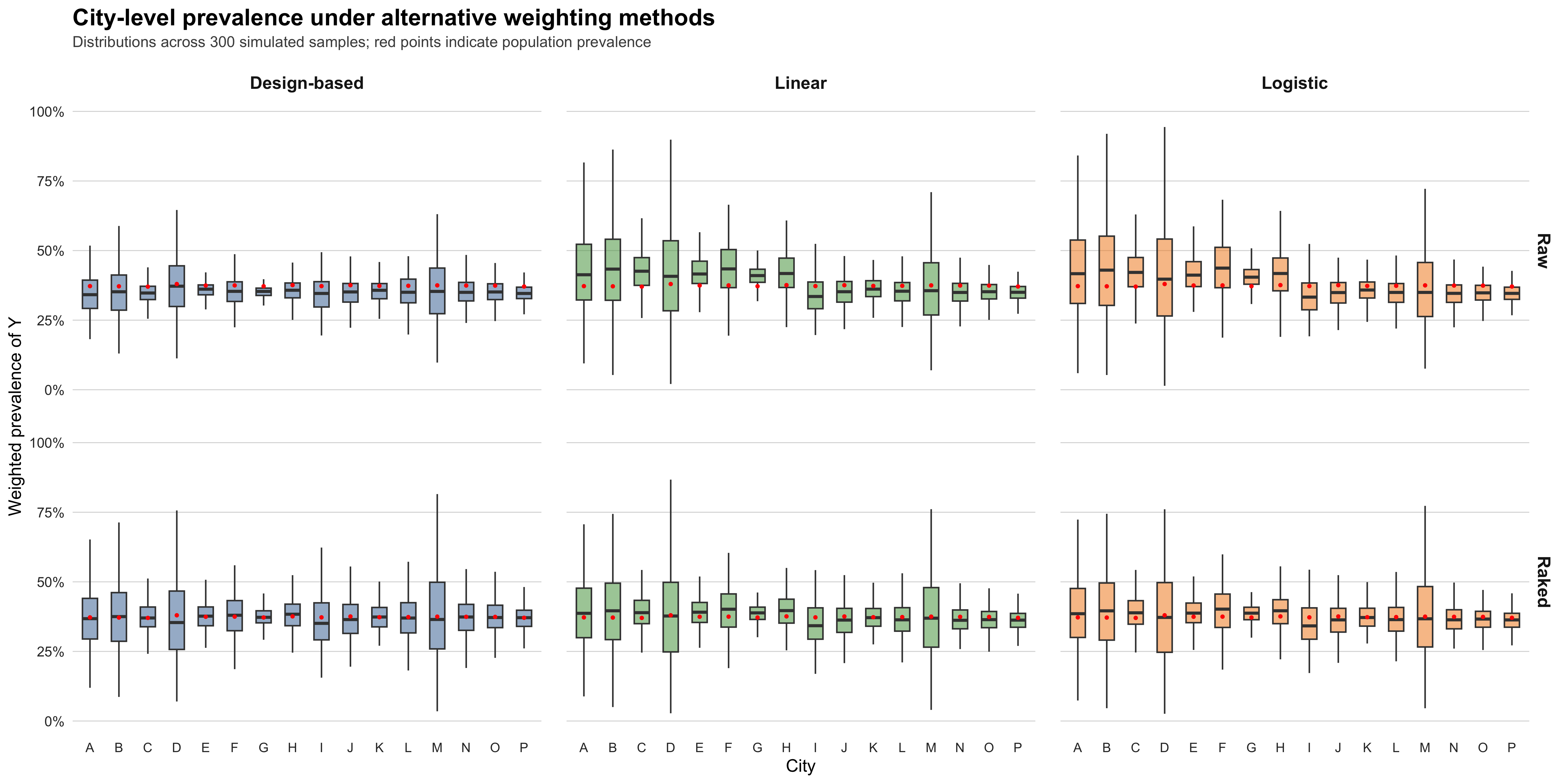}
    \caption{City-level prevalence across 300 simulated samples under
    alternative weighting methods. Columns correspond to the
    design-based, linear, and logistic weighting methods. The upper
    row
    shows estimates based on raw weights, while the lower row shows
    estimates after raking. Red points indicate the corresponding
    population prevalence for each city.}
    \label{fig:city_prevalence}
\end{figure}

The city-level results provide a complementary assessment of the overall
population-mean results in Table~\ref{tab:simulation_performance}. The
red points indicate the corresponding population prevalence for each
city, allowing the location and dispersion of the sampling distributions
to be compared across weighting methods. The raw linear and logistic
estimators exhibit greater between-sample variability than the
design-based estimator for several cities, with the variability being
particularly pronounced for the logistic estimator in some subgroups.
At the same time, the centers of the model-based distributions are
generally close to the corresponding population prevalences. After
raking, the distributions become more concentrated around the city-level
population prevalences for most cities, although the magnitude of this
change varies across cities and weighting methods.

To further connect the finite-sample evidence with the asymptotic theory, we
provide an additional check of Theorem~1. The theorem establishes that the
effective sample size ratio satisfies
\[
\frac{\mathrm{ESS}_n}{n}
\overset{p}{\longrightarrow}
\frac{
\left\{\mathbb{E}\left[
w^{\mathrm{raw}}(X;\beta_0,\pi)
\right]\right\}^2
}{
\mathbb{E}\left[
w^{\mathrm{raw}}(X;\beta_0,\pi)^2
\right]
}.
\]
Under the finite-cell framework, the population limit can be written as
\[
\frac{
\left\{
\sum_{s=1}^S
\pi_s w^{\mathrm{raw}}(x_s;\beta_0,\pi)
\right\}^2
}{
\sum_{s=1}^S
\pi_s
\left[w^{\mathrm{raw}}(x_s;\beta_0,\pi)\right]^2
},
\qquad
\pi_s = \Pr(X=x_s).
\]

Using 300 simulation replications, we compare the empirical ESS ratio
computed from the estimated individual equivalent weights,
\[
\frac{\bar{\hat w}^{\,2}}
{\overline{\hat w^2}},
\]
with a cell-based finite-sample approximation to the population limit,
\[
\frac{
\left\{
\sum_{s=1}^S
\hat\pi_s w^{\mathrm{raw}}(x_s;\beta_0,\pi)
\right\}^2
}{
\sum_{s=1}^S
\hat\pi_s
\left[w^{\mathrm{raw}}(x_s;\beta_0,\pi)\right]^2
},
\qquad
\hat\pi_s=\frac{n_s}{n}.
\]
The latter uses the population-limit equivalent weights while replacing
the population cell proportions $\pi_s$ by their empirical counterparts
$\hat\pi_s$. Since $\hat\pi_s \overset{p}{\to} \pi_s$, this cell-based quantity
converges to the population ESS limit in Theorem~1.

Table~\ref{tab:ess-mean-iqr} reports the mean and interquartile range of the
two quantities across the 300 replications. The empirical ESS ratio and the
cell-based approximation are similar, providing finite-sample evidence
consistent with the convergence established in Theorem~1.

\renewcommand{\arraystretch}{1.3}
\begin{table}[H]
\centering
\caption{Finite-sample ESS/$n$ ratios: mean and IQR}
\label{tab:ess-mean-iqr}
\begin{tabular}{lcc}
\hline
 & Mean & IQR [Q1, Q3] \\
\hline
Empirical ESS ratio
& 0.371 & [0.360, 0.385] \\

Cell-based approximation
& 0.379 & [0.369, 0.393] \\
\hline
\end{tabular}
\end{table}
\renewcommand{\arraystretch}{1.0}

Finally, to assess the asymptotic linearity result in
Theorem~\ref{thm:AL}, we examine the finite-sample
distribution of the centered and scaled equivalent-weight estimator.
For each replication, let $\hat{\theta}_{\mathrm{EW}}$ denote the
logistic equivalent-weight estimator and let
$\theta_{\mathrm{EW}}=T(\beta_0,\pi)$ denote its population-limit target.
Theorem~\ref{thm:AL} establishes that
\[
\sqrt{n}\left(
\hat{\theta}_{\mathrm{EW}}-\theta_{\mathrm{EW}}
\right)
\Rightarrow
N(0,V),
\]
where
\[
V=\mathbb{E}\left[\psi(Z)^2\right],
\]
and $\psi(Z)$ is the influence function given in
Theorem~\ref{thm:AL}.

We therefore assess the corresponding standardized statistic
\[
Z_n
=
\frac{
\sqrt{n}\left(
\hat{\theta}_{\mathrm{EW}}-\theta_{\mathrm{EW}}
\right)
}{
\sqrt{\widehat V}
},
\]
which should be approximately standard normal in finite samples. A
normal Q--Q plot of $Z_n$ is used to assess the adequacy of this
asymptotic normal approximation. The Q--Q plot of the standardized estimator shows close agreement with the standard normal distribution over the central portion of the
distribution, with modest deviations in the tails. This provides evidence that the normal approximation underlying the Wald inference is reasonable in the simulated setting.

\begin{figure}[H]
\centering
\begin{minipage}{0.6\textwidth}
    \centering
    \includegraphics[width=\textwidth]{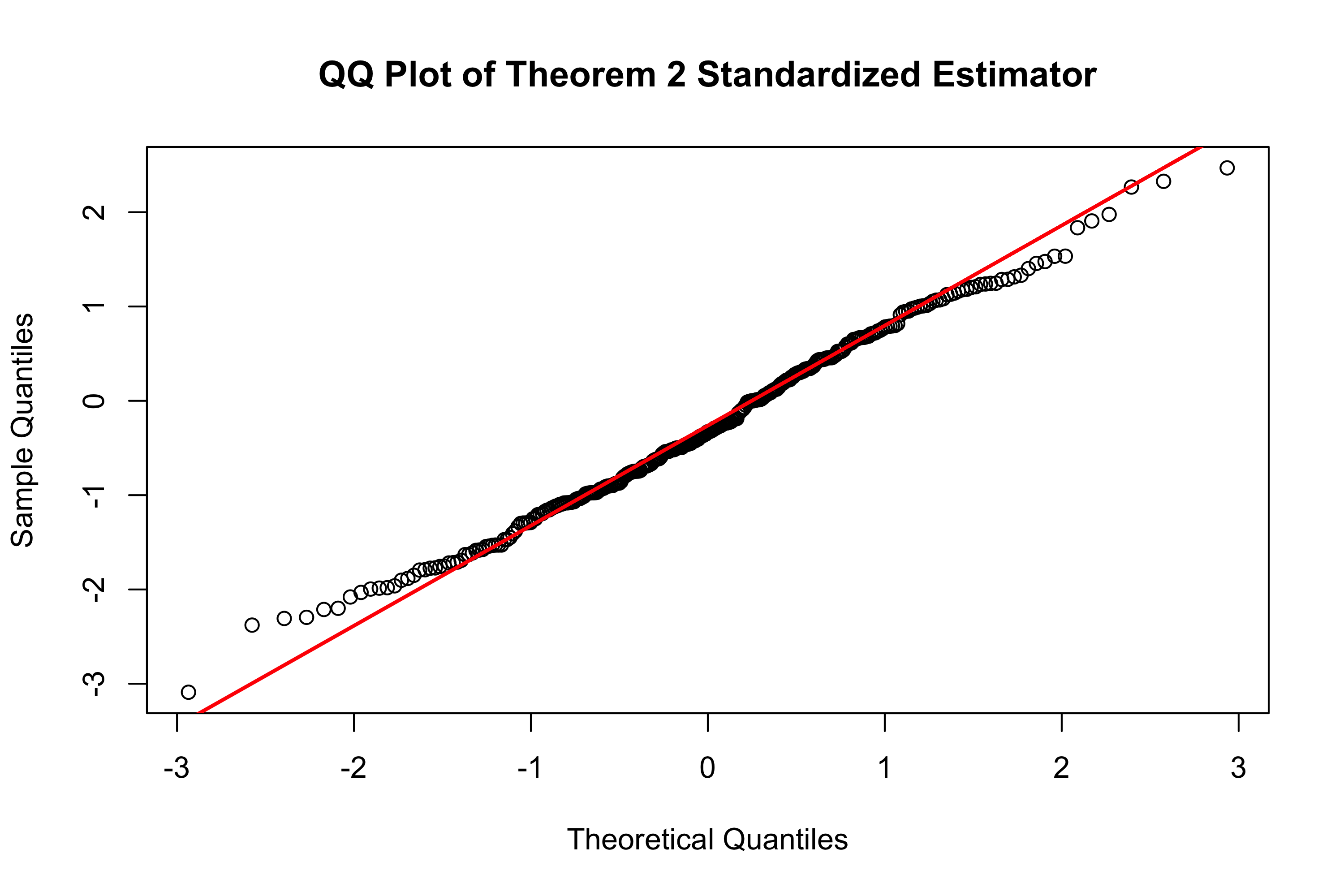}
    \caption{Normal Q--Q plot of the standardized equivalent-weight
estimator $\sqrt{n}(\hat{\theta}_{\mathrm{EW}}-\theta_{\mathrm{EW}})$
across 300 replications.}
    \label{fig:qqplot}
\end{minipage}
\end{figure}

We next evaluate the plug-in variance estimator from
Theorem~2. For asymptotic variance
estimation, we use the unnormalized raw equivalent weights when
computing the gradient term $\hat A$, while the equivalent weights are normalized to sum to n in the weighted-mean component of the influence function. This normalization ensures
that $\hat\theta$ is expressed on the $1/n$ scale required in
Theorem~\ref{thm:AL}. Across the 300 replications, the estimated asymptotic
variance has mean approximately $0.56$ and interquartile range
approximately $[0.53,0.59]$. The empirical Monte Carlo asymptotic
variance is approximately $0.54$, indicating close agreement between
the plug-in variance estimator and the observed sampling variability.

For comparison, we also construct Wald confidence intervals for the
linear equivalent-weight estimator using its standard influence-function
variance estimator. For the linear equivalent-weight estimator, we estimate its asymptotic
variance using the corresponding influence function. Let
$\hat{\boldsymbol{\beta}}$ denote the OLS coefficient estimator,
$\hat{\mathbf Q}=n^{-1}\mathbf X^\top\mathbf X$, and let
$\hat{\mathbf q}$ denote the population mean of the design variables.
The estimated influence function for observation $i$ is
\[
\hat{\psi}_i
=
\mathbf{x}_i^\top
\hat{\mathbf Q}^{-1}
\hat{\mathbf q}
\left(y_i-\mathbf{x}_i^\top\hat{\boldsymbol{\beta}}\right),
\]
giving the plug-in estimator of the asymptotic variance
\[
\widehat{V}_{\mathrm{lin}}
=
\frac{1}{n}
\sum_{i=1}^n
\hat{\psi}_i^2.
\]

Table~\ref{tab:variance_coverage} reports the
empirical Monte Carlo variance, mean estimated variance, mean standard
error, and empirical coverage of nominal $95\%$ confidence intervals.

\begin{table}[H]
\centering
\caption{Variance estimation and confidence-interval performance}
\label{tab:variance_coverage}
\begin{tabular}{lrrrr}
\hline
Method & Empirical $V$ & Mean $\hat V$ & Mean SE
& 95\% Coverage \\
\hline
Linear   & 0.437 & 0.479 & 0.0126 & 97.6\% \\
Logistic & 0.542 & 0.562 & 0.0136 & 94.7\% \\
\hline
\end{tabular}
\end{table}

For the logistic estimator, the mean plug-in variance estimate
($0.562$) is close to the empirical Monte Carlo variance ($0.542$).
The resulting nominal $95\%$ Wald confidence intervals achieve
$94.7\%$ empirical coverage, indicating approximately nominal
finite-sample inference in this simulation. The linear estimator also
shows good coverage at $97.6\%$, although its plug-in variance estimate
is somewhat more conservative in this setting.

\subsection{Applications to FFCWS Study}

We apply the weighting methods to the FFCWS using $n=3,442$ sampled
units representing a population of 1,131,308 individuals. The linear and
logistic equivalent weights are first constructed in their raw form,
without imposing positivity constraints. Because equivalent weights are
interpreted as local sensitivity coefficients rather than sampling
probabilities, negative values are retained in the raw model-based
estimators. For comparisons involving raking, a separate
positivity-adjustment step is required because the raking procedure
requires positive starting weights. The resulting raked weights are
therefore treated as calibrated versions of the raw equivalent weights
rather than as the equivalent weights themselves.

The logistic model uses father interview status as the binary response
for constructing the weights because this variable has no missing values.
This response is used only to construct the weights; the substantive
outcome analyzed below is CPS contact. Recent work using the FFCWS has
documented underreporting of CPS involvement in self-reported measures
and examined its implications for analyses of young-adult outcomes
(Berger et al., \citeyear{BergerEtAl2026}). We use CPS contact here as the substantive outcome
for illustrating the proposed weighting and inference framework. When
raking the model-based weights, we use the same four demographic
variables as in the national FFCWS weights, excluding city. We refer to
the national FFCWS weights as the design-based weights; these weights are
already calibrated in the FFCWS data.

The population-count results in Table~\ref{tab:ffcws_births} show that
the linear and logistic model-based procedures reproduce the overall
population total exactly after normalization, but allocate the total
differently across education categories. These differences illustrate
how the model-based approaches use the observed relationships between
the weighting variables and the model outcome to distribute population
mass across poststratification cells, rather than reproducing the
national design-based allocation exactly.

Table~\ref{tab:ffcws_ess} shows that the raw logistic equivalent weights
have the largest ESS among the model-based methods, with an ESS of
1,435 compared with 1,273 for the linear equivalent weights. After
raking, the logistic weights likewise have the largest ESS among the
raked methods, with an ESS of 1,301 compared with 1,082 for the
raked linear weights, while the raked design-based weights have the
lowest ESS. Thus, in this application, the logistic construction
achieves a relatively large effective sample size both before and
after raking.

\begin{table}[H]
\centering
\caption{Population counts of live births by education}
\label{tab:ffcws_births}
\begin{tabular}{lrrrrrr}
\hline
Education      & $<$8 grade & Some HS & HS & Some College & College+ & Total \\
\hline
Design-raked   & 111,324 & 211,988 & 340,211 & 214,319 & 253,467 & 1,131,309 \\
Linear model   & 73,232  & 267,083 & 367,476 & 242,049 & 181,469 & 1,131,309 \\
Logistic model & 60,214 & 207,301 & 316,517 & 258,784 & 288,491 & 1,131,309 \\
\hline
\end{tabular}
\end{table}

\begin{table}[H]
\centering
\caption{Effective Sample Size (ESS) for FFCWS Weights ($n=3442$)}
\label{tab:ffcws_ess}
\begin{tabular}{lccccc}
\hline
 & Design-based & Linear & Logit & Linear-raked & Logit-raked \\
\hline
ESS & 527 & 1273 & 1435 & 1082 & 1301 \\
\hline
\end{tabular}
\end{table}

Using the logistic equivalent weights, the estimated prevalence of CPS
contact in the FFCWS is $\hat\theta=0.49$ with plug-in variance estimate
$\hat V=0.42$, yielding a 95\% confidence interval of $[0.47,0.51]$
under the proposed variance calculation. This illustrates how the
proposed variance estimator can be applied to a single survey dataset
without bootstrap or replication-based variance estimation. Because the
weighting response and substantive outcome are different variables in
this application, the estimate should be interpreted primarily as an
illustration of the weighting and inference framework rather than as
evidence that father interview status is itself an adequate predictive
model for CPS-contact prevalence.

To examine subgroup behavior, we compare CPS-contact prevalence across
cities in Figure~\ref{fig:ffcws_cps}. The weighting methods produce
different adjustments across cities, with the magnitude of the change
depending on the weighting method and the city. The design-based and
linear weights produce relatively similar patterns across cities,
whereas the logistic weights produce larger adjustments for several
cities. Thus, the application illustrates that the choice of weighting
method can have heterogeneous effects across population subgroups,
even when the overall population estimate is relatively stable.

\begin{figure}[H]
\centering
\includegraphics[width=0.95\textwidth]{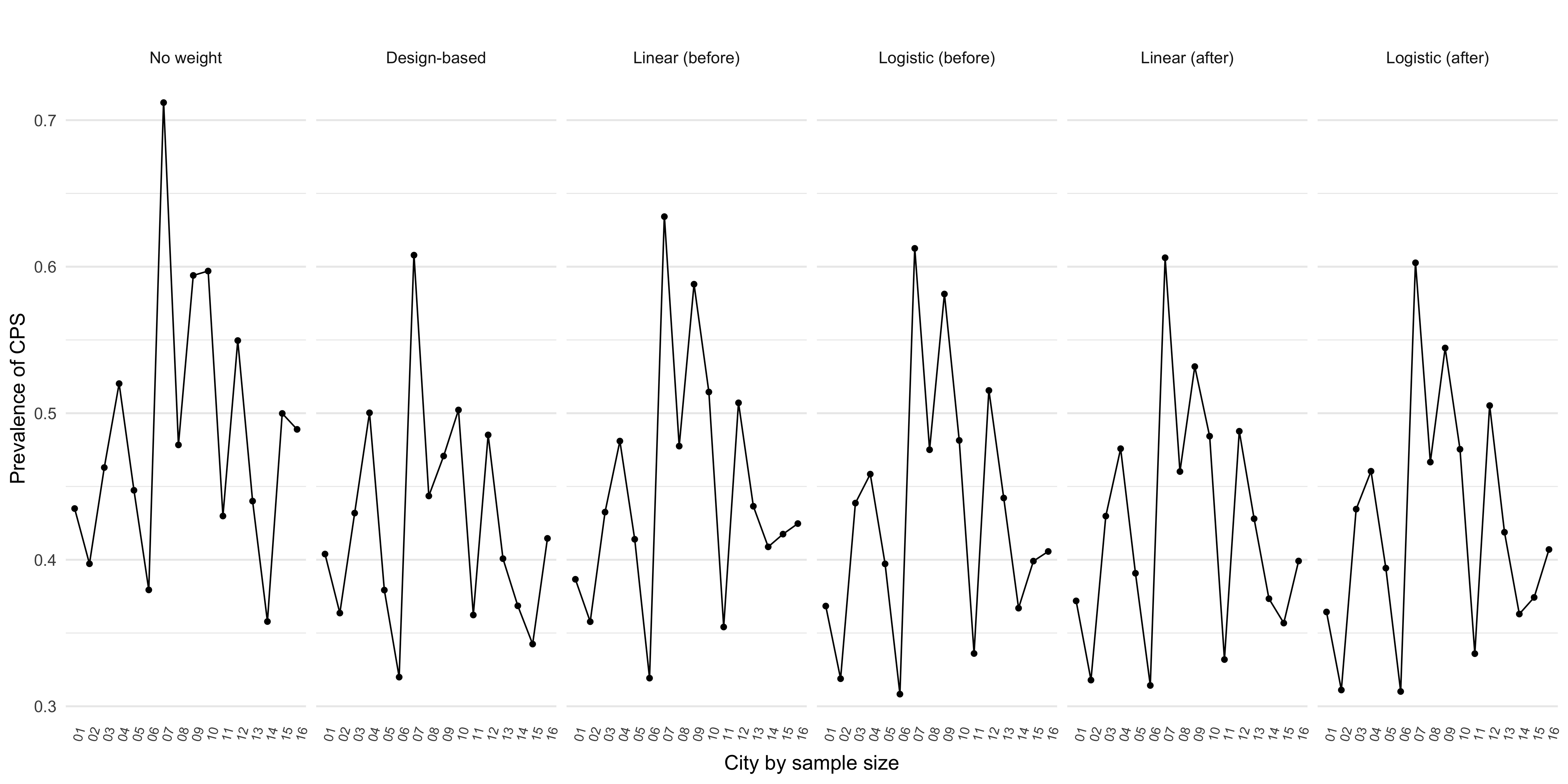}
\caption{Prevalence of Child Protective Services (CPS) contact by city
under alternative weighting methods. Estimates are shown before and
after raking.}
\label{fig:ffcws_cps}
\end{figure}

\section{Discussion}
\subsection{Key Findings}

This study develops a frequentist equivalent-weight construction for
logistic regression under categorical poststratification. The central
contribution is to represent the nonlinear model-based population
estimator locally through weights defined as scaled derivatives with
respect to the observed outcomes. Specifically, starting from the
logistic estimating equations, we use implicit differentiation of the
score equation, followed by the chain rule for the model-based
population estimator, to obtain
\[
\hat w_i
=
n\frac{\partial\hat\theta^{LR}}{\partial y_i}.
\]
This provides a direct interpretation of the equivalent weights as local
sensitivity coefficients and yields a closed-form expression that can
be evaluated from the fitted logistic model and the population cell
structure.

This frequentist development is closely related to the locally
equivalent-weight framework of Giordano et al.~(\citeyear{GiordanoEtAl2026}), whose Bayesian
MrP formulation leads to a corresponding closed-form expression under
categorical poststratification. The present work focuses instead on the
frequentist survey-weighting interpretation and inference, including
effective sample size, asymptotic linearity, plug-in variance estimation,
and finite-sample evaluation.

The simulation results provide several findings regarding the practical
behavior of the proposed construction. First, both linear and logistic
model-based weighting achieve lower regression-coefficient MSE than the
design-based and unweighted approaches in the simulated setting,
although the advantage is not uniform across replications. Thus,
incorporating outcome information through a regression model can improve
regression performance under the data-generating process considered.

For overall population-mean estimation, the linear equivalent-weight
estimator performs best in the reported simulation, with an MSE of
$0.000148$ and essentially no bias. The logistic estimator has a
slightly larger MSE of $0.000194$ and a small bias of $-0.004$, but
remains competitive with the linear estimator. These results do not
support a claim that logistic equivalent weighting uniformly improves
point estimation. Rather, they demonstrate that the nonlinear logistic
construction can provide population estimates comparable to established
weighting approaches while allowing a nonlinear outcome model to be
incorporated into the equivalent-weight framework.

A second useful property of the proposed construction is that the
resulting weights can be examined using effective sample size. In the
simulation, the model-based equivalent weights, particularly the
logistic weights, show relatively stable ESS performance before and
after raking. Before raking, the logistic weights have the largest ESS,
although their ESS is very close to that of the linear equivalent
weights. After raking, the linear and logistic weights remain similar
in ESS, with the linear weights yielding a slightly larger ESS in this
simulation. The cell-based and empirical ESS ratios are also close to
one another, providing finite-sample evidence consistent with the
population ESS limit established theoretically. For the simulated
sample size of $n=3000$, the empirical ESS ratio of approximately
$0.38$ corresponds to an effective sample size of approximately $1,140$.

The simulation also provides evidence supporting the proposed
frequentist inference. The logistic estimator has an empirical Monte
Carlo variance of $0.542$ compared with a mean plug-in variance estimate
of $0.562$. The resulting Wald confidence intervals achieve $94.7\%$
empirical coverage for a nominal $95\%$ confidence level. Together with
the approximately Gaussian behavior of the centered and scaled
estimator, these results suggest that the asymptotic linear
representation and plug-in variance estimator provide a useful
finite-sample approximation at the sample size considered.

The city-level analysis provides an important complement to the overall
population-mean comparison. Although the linear estimator has the
smallest overall MSE in the simulation, the city-level results suggest
that model-based weighting can provide useful local population recovery
across heterogeneous subgroups. The logistic estimator often places the
center of the city-level sampling distribution close to the corresponding
population prevalence, but this improvement in local recovery can be
accompanied by greater between-sample variability in some cities. This
suggests a local bias--variance tradeoff: the flexibility of the
nonlinear logistic model can improve recovery for particular subgroups,
while the resulting equivalent weights may also be more sensitive to
the realized sample composition within those subgroups. Raking generally
reduces this variability and brings the city-level estimates into closer
alignment with their population benchmarks.

\subsection{Limitations}

Several limitations should be emphasized. First, the proposed logistic
equivalent weights provide a local representation of a nonlinear
model-based population estimator rather than an exact global
representation by a fixed weighted sample mean. The weights therefore
describe first-order sensitivity at the observed data and may not
characterize the behavior of the estimator under large changes in the
outcomes. The local nature of the representation is intrinsic to the
nonlinear logistic model rather than a consequence of an estimation
approximation that can simply be removed.

Second, the equivalent weights can be negative or highly variable.
Negative values are mathematically permissible because the weights are
local sensitivity coefficients rather than sampling probabilities.
However, such weights can complicate their use in procedures designed
for conventional positive survey weights, including some calibration
and raking procedures. In particular, positivity constraints may require
a separate calibration or transformation step if the equivalent weights
are to be used as starting weights for such procedures. Any such
transformation defines a different estimator and is not covered by the
theoretical results developed here. The present theory therefore concerns
the raw equivalent weights and should not be interpreted as providing
inference for arbitrary positivity-constrained or calibrated versions.

Third, the validity of the logistic construction depends on the
specification of the outcome model and on the choice of weighting
variables. The simulations considered here do not provide a comprehensive
assessment of model misspecification, omitted interactions, nonlinear
effects, or sparse poststratification cells. In settings where the
logistic model poorly describes the relationship between the outcome and
the weighting variables, the resulting population estimator and its
equivalent weights may perform differently from the results reported
here. A systematic comparison under model misspecification is therefore
an important direction for future work.

Fourth, the asymptotic theory is developed under a finite-cell
superpopulation framework rather than a general design-based theory for
complex surveys. In particular, the current results do not explicitly
account for clustering, unequal-probability multistage sampling, or
design effects arising from complex sample designs. Extending the
theoretical results to these settings would broaden the applicability
of the proposed construction.

Fifth, the method requires population information for the
poststratification cells. The quality of the resulting population
estimator therefore depends on the availability and accuracy of the
population cell counts and on the choice of auxiliary variables. The
method does not automatically incorporate auxiliary information that is
not represented in the specified poststratification structure.

Finally, the FFCWS application should be interpreted primarily as an
illustration of the weighting and variance-estimation framework. The
logistic weights are constructed using father interview status, whereas
the substantive outcome is CPS contact. Consequently, the application
does not by itself establish that the weighting model is optimal for
estimating CPS-contact prevalence. Rather, it demonstrates how the
proposed equivalent weights and plug-in variance estimator can be
implemented in an empirical survey setting.

\subsection{Future Research}

Several extensions are natural. First, the theoretical results could be generalized to complex survey designs, including stratified multistage sampling, clustering, and unequal-probability sampling. Second, a systematic study of model misspecification could clarify when logistic weighting provides an advantage over linear or design-based approaches. Third, methods for stabilizing extreme and negative weights could be developed with explicit calibration and inferential guarantees.

A further direction is to incorporate richer population information through hierarchical and nonparametric models. In particular, Gaussian-process and hierarchical Bayesian approaches can provide flexible representations of population-level relationships while allowing uncertainty about sparse cells to be propagated through the weighting procedure. Such extensions provide a natural direction for future work.

\appendix

\section{Supplementary One-Unit Perturbation Analysis}
\label{app:perturbation}

This appendix provides a supplementary numerical assessment of the
sensitivity of the logistic equivalent weights to individual outcome
perturbations. This analysis is not required for the derivation of the
equivalent weights in Section~2 or for the asymptotic results. Instead,
it examines the extent to which the fitted weights may change when a
single observed binary outcome is perturbed.

\subsection{One-Unit Perturbation Analysis}

Let $y^{(-i)}$ denote the response vector obtained by changing an
observed value $y_i=1$ to $y_i=0$, while leaving all other outcomes
unchanged. Equivalently,
\[
y^{(-i)}=y-e_i,
\]
where $e_i=(0,\ldots,0,1,0,\ldots,0)^\top$ has its $1$ in position $i$.

For each perturbation, the logistic regression model is re-estimated
using $y^{(-i)}$, and the equivalent weights are recomputed. Denote the
resulting weight vector by
\[
\hat w^{(-i)}=\hat w(y^{(-i)}),
\]
while $\hat w=\hat w(y)$ denotes the weights obtained from the original
observed response vector.

A first-order finite-difference approximation to the derivative of
the weight vector with respect to $y_i$ is
\begin{equation}
\frac{\partial\hat w(y)}{\partial y_i}
\approx
\frac{\hat w(y^{(-i)})-\hat w(y)}
{y_i^{(-i)}-y_i}
=
\hat w(y)-\hat w(y^{(-i)}),
\label{eq:app-weight-derivative}
\end{equation}
because $y_i^{(-i)}-y_i=-1$. Thus, the change in the fitted weight
vector resulting from changing one observed outcome from $1$ to $0$
provides a direct numerical approximation to its local outcome
sensitivity.

Because $y$ is binary, the Jacobian--vector product can be written as
\[
(\nabla\hat w)y
=
\sum_{i:y_i=1}
\frac{\partial\hat w(y)}{\partial y_i}.
\]

Using the one-unit perturbations in \eqref{eq:app-weight-derivative}
gives the finite-difference approximation
\begin{equation}
(\nabla\hat w)y
\approx
\sum_{i:y_i=1}
\left\{
\hat w(y)-\hat w(y^{(-i)})
\right\}.
\label{eq:app-finite-difference-correction}
\end{equation}

This calculation provides a direct numerical assessment of the
magnitude of the outcome-dependent component of the fitted weights.

\subsection{Empirical Sensitivity Assessment}

Across the simulation experiments, the outcome-dependent correction term $(\nabla\hat w)y$ was extremely small. The observed values ranged from approximately
$-1.4\times10^{-4}$ to $4.2\times10^{-5}$, with mean approximately
$-4.0\times10^{-6}$ and median approximately $-2.4\times10^{-6}$.

In a separate one-unit perturbation check, changing a single binary
outcome altered
\[
(\nabla\hat w)y+\hat w
\]
by at most approximately $0.02$, with a median change of approximately
$0.002$.

These calculations provide a numerical sensitivity check for the
settings considered. They suggest that the indirect dependence of the
fitted equivalent weights on individual outcomes through the fitted
logistic model is small in the simulation and empirical settings
examined. They are not intended to establish that outcome dependence
is asymptotically negligible under all sampling designs or model
configurations, and they are not required for the derivation of the
equivalent weights.

\section{Proofs}

\singlespacing

\subsection{Proof of Lemma~\ref{lem:weight-regularity}}

\begin{proof}
\setlength{\abovedisplayskip}{4pt plus 1pt minus 1pt}
\setlength{\belowdisplayskip}{4pt plus 1pt minus 1pt}
\setlength{\abovedisplayshortskip}{2pt plus 1pt minus 1pt}
\setlength{\belowdisplayshortskip}{3pt plus 1pt minus 1pt}
\setlength{\jot}{2pt}

Recall that
\[
w^{\mathrm{raw}}(x;\beta,\tilde\pi)
=
x^\top Q(\beta,\tilde\pi)^{-1}C(\beta),
\]
where
\[
Q(\beta,\tilde\pi)
=
\sum_{s=1}^S
\tilde\pi_s v_s(\beta)x_sx_s^\top,
\qquad
C(\beta)
=
\sum_{s=1}^S
q_s v_s(\beta)x_s.
\]
Since the number of poststratification cells is finite and
$\mathcal X=\{x_1,\ldots,x_S\}$, there exists a finite constant
$M_x$ such that
\[
\max_{1\le s\le S}\|x_s\|\le M_x.
\]
Moreover, the logistic variance function satisfies
\[
0<v_s(\beta)
=
\sigma(x_s^\top\beta)
\{1-\sigma(x_s^\top\beta)\}
\le \frac14.
\]
Because the cell vectors and target proportions are fixed and
$\mathcal N$ is a fixed neighborhood, the vector $C(\beta)$ is
uniformly bounded on $\mathcal N$. Thus, for some finite constant
$C_C$,
\[
\sup_{\beta\in\mathcal N}\|C(\beta)\|
\le C_C.
\]

By assumption, the smallest eigenvalue of $Q(\beta,\tilde\pi)$ is
uniformly bounded away from zero on $\mathcal N\times\Pi$:
\[
\lambda_{\min}\{Q(\beta,\tilde\pi)\}\ge c>0
\]
for every $(\beta,\tilde\pi)\in\mathcal N\times\Pi$. Since
$Q(\beta,\tilde\pi)$ is symmetric positive definite, its inverse
therefore satisfies
\[
\|Q(\beta,\tilde\pi)^{-1}\|_{\mathrm{op}}
=
\frac{1}{
\lambda_{\min}\{Q(\beta,\tilde\pi)\}
}
\le \frac1c.
\]
Combining these bounds yields
\[
\begin{aligned}
\left|
w^{\mathrm{raw}}(x;\beta,\tilde\pi)
\right|
&\le
\|x\|\,
\|Q(\beta,\tilde\pi)^{-1}\|\,
\|C(\beta)\| \\
&\le
M_xc^{-1}C_C.
\end{aligned}
\]
Consequently, there exists $C<\infty$ such that
\[
\sup_{(\beta,\tilde\pi)\in\mathcal N\times\Pi}
\max_{x\in\mathcal X}
\left|
w^{\mathrm{raw}}(x;\beta,\tilde\pi)
\right|
\le C.
\]
This establishes the uniform envelope.

We next establish continuity. Since the logistic function is
continuously differentiable, $v_s(\beta)$ is continuous in $\beta$.
It follows that $Q(\beta,\tilde\pi)$ and $C(\beta)$ are continuous in
$(\beta,\tilde\pi)$. The uniform lower bound on the eigenvalues of
$Q(\beta,\tilde\pi)$ ensures that the matrix remains nonsingular
throughout $\mathcal N\times\Pi$, so its inverse is continuous there.
Hence the mapping
\[
(\beta,\tilde\pi)
\mapsto
w^{\mathrm{raw}}(x;\beta,\tilde\pi)
\]
is continuous for every $x\in\mathcal X$.

Finally, for each cell $s$, the law of large numbers gives
\[
\hat\pi_s
=
\frac1n\sum_{i=1}^n I(X_i=x_s)
\xrightarrow{p}\pi_s.
\]
Because the number of cells is finite, these componentwise convergence
results imply $\hat\pi\xrightarrow{p}\pi$. Together with the uniform
envelope established above, this yields the stated $O_p(1)$ envelope
for the empirical weights, uniformly in $\beta$.
\end{proof}

\subsection{Proof of Theorem~\ref{thm:plugin-LLN}}

\begin{proof}
\setlength{\abovedisplayskip}{4pt plus 1pt minus 1pt}
\setlength{\belowdisplayskip}{4pt plus 1pt minus 1pt}
\setlength{\abovedisplayshortskip}{2pt plus 1pt minus 1pt}
\setlength{\belowdisplayshortskip}{3pt plus 1pt minus 1pt}
\setlength{\jot}{2pt}

For notational simplicity, let
\[
w_i(\beta,\tilde\pi)
=
w^{\mathrm{raw}}(X_i;\beta,\tilde\pi).
\]
Since the number of poststratification cells is finite, the law of
large numbers gives $\hat\pi\xrightarrow{p}\pi$. Together with the
consistency assumption $\hat\beta\xrightarrow{p}\beta_0$ and
Lemma~\ref{lem:weight-regularity}, this implies
\[
\sup_{1\le i\le n}
\left|
w_i(\hat\beta,\hat\pi)
-
w_i(\beta_0,\pi)
\right|
\xrightarrow{p}0.
\]
Consequently, the difference between the empirical averages evaluated
at the estimated and limiting arguments converges to zero in
probability:
\[
\frac1n\sum_{i=1}^n
w_i(\hat\beta,\hat\pi)
-
\frac1n\sum_{i=1}^n
w_i(\beta_0,\pi)
\xrightarrow{p}0.
\]
By the law of large numbers,
\[
\frac1n\sum_{i=1}^n
w_i(\beta_0,\pi)
\xrightarrow{p}
E\left[
w^{\mathrm{raw}}(X;\beta_0,\pi)
\right].
\]
Combining these two convergence results yields
\[
\bar w_n(\hat\beta,\hat\pi)
\xrightarrow{p}
E\left[
w^{\mathrm{raw}}(X;\beta_0,\pi)
\right].
\]

An analogous argument establishes convergence of the empirical second
moment. In particular, the uniform envelope in
Lemma~\ref{lem:weight-regularity} ensures that the squared weights are
uniformly bounded, and hence
\[
\overline{w^2}_n(\hat\beta,\hat\pi)
\xrightarrow{p}
E\left[
w^{\mathrm{raw}}(X;\beta_0,\pi)^2
\right].
\]
Recalling the representation of the effective sample size,
\[
\frac{\mathrm{ESS}_n}{n}
=
\frac{
\bar w_n(\hat\beta,\hat\pi)^2
}{
\overline{w^2}_n(\hat\beta,\hat\pi)
},
\]
the continuous mapping theorem therefore gives
\[
\frac{\mathrm{ESS}_n}{n}
\xrightarrow{p}
\frac{
\left[
E\{w^{\mathrm{raw}}(X;\beta_0,\pi)\}
\right]^2
}{
E\left[
w^{\mathrm{raw}}(X;\beta_0,\pi)^2
\right]
}.
\]
This establishes the claimed convergence.
\end{proof}

\subsection{Proof of Lemma~\ref{lem:weight-differentiability}}

\begin{proof}
\setlength{\abovedisplayskip}{4pt plus 1pt minus 1pt}
\setlength{\belowdisplayskip}{4pt plus 1pt minus 1pt}
\setlength{\abovedisplayshortskip}{2pt plus 1pt minus 1pt}
\setlength{\belowdisplayshortskip}{3pt plus 1pt minus 1pt}
\setlength{\jot}{2pt}

For notational simplicity, write
\[
Q=Q(\beta,\tilde\pi),
\qquad
C=C(\beta).
\]
The logistic variance function
\[
v_s(\beta)
=
\sigma(x_s^\top\beta)
\{1-\sigma(x_s^\top\beta)\}
\]
is continuously differentiable, with gradient
\[
\nabla_\beta v_s(\beta)
=
v_s(\beta)
\{1-2\sigma(x_s^\top\beta)\}
x_s.
\]
Since $\mathcal X$ is finite and $v_s(\beta)\le1/4$, these
derivatives are uniformly bounded on $\mathcal N$. Consequently,
$Q(\beta,\tilde\pi)$ and $C(\beta)$ are continuously differentiable
in $(\beta,\tilde\pi)$.

Moreover, the uniform nonsingularity of $Q(\beta,\tilde\pi)$ ensures
that its inverse is continuously differentiable, with derivative
\[
D(Q^{-1})[\Delta]
=
-Q^{-1}\Delta Q^{-1}.
\]
For a direction $h$ in $\beta$, the corresponding directional
derivatives of $Q$ and $C$ are
\[
DQ(\beta,\tilde\pi)[h]
=
\sum_{s=1}^S
\tilde\pi_s
\{\nabla_\beta v_s(\beta)^\top h\}
x_sx_s^\top,
\]
and
\[
DC(\beta)[h]
=
\sum_{s=1}^S
q_s
\{\nabla_\beta v_s(\beta)^\top h\}
x_s.
\]
Since there are finitely many cells and all cell vectors and
probabilities are bounded, the derivative of
$w^{\mathrm{raw}}(x;\beta,\tilde\pi)$ with respect to $\beta$ is
uniformly bounded. In particular, there exists a finite constant
$C_\beta$ such that
\[
\sup_{(\beta,\tilde\pi)\in\mathcal N\times\Pi}
\left\|
\nabla_\beta
w^{\mathrm{raw}}(x;\beta,\tilde\pi)
\right\|
\le C_\beta.
\]

Similarly, because $Q$ is linear in $\tilde\pi$, its directional
derivative in a direction $h_\pi$ is
\[
DQ(\beta,\tilde\pi)[h_\pi]
=
\sum_{s=1}^S
h_{\pi,s}
v_s(\beta)x_sx_s^\top,
\]
whereas $C(\beta)$ does not depend on $\tilde\pi$. Applying the
derivative formula for the matrix inverse therefore gives
\[
D w^{\mathrm{raw}}(x;\beta,\tilde\pi)[h_\pi]
=
-x^\top
Q^{-1}
DQ[h_\pi]
Q^{-1}
C.
\]
The finite number of cells, bounded cell vectors, and uniform bound
on $\|Q^{-1}\|$ consequently imply that there exists a finite
constant $C_\pi$ such that
\[
\sup_{(\beta,\tilde\pi)\in\mathcal N\times\Pi}
\left\|
\nabla_\pi
w^{\mathrm{raw}}(x;\beta,\tilde\pi)
\right\|
\le C_\pi.
\]
Thus, the population-limit raw equivalent weight is continuously
differentiable in both $\beta$ and $\tilde\pi$.

Finally, define
\[
T(\beta,\tilde\pi)
=
\frac{
P[w^{\mathrm{raw}}(X;\beta,\tilde\pi)Y]
}{
P[w^{\mathrm{raw}}(X;\beta,\tilde\pi)]
}.
\]
Because the limiting equivalent weights have a nonzero population
mean, the denominator is bounded away from zero in a sufficiently
small neighborhood of $(\beta_0,\pi)$. The quotient rule therefore
implies that $T(\beta,\tilde\pi)$ is continuously differentiable in
both arguments.
\end{proof}

\subsection{Proof of Theorem~\ref{thm:AL}}

\begin{proof}
\setlength{\abovedisplayskip}{4pt plus 1pt minus 1pt}
\setlength{\belowdisplayskip}{4pt plus 1pt minus 1pt}
\setlength{\abovedisplayshortskip}{2pt plus 1pt minus 1pt}
\setlength{\belowdisplayshortskip}{3pt plus 1pt minus 1pt}
\setlength{\jot}{2pt}

For clarity, write
\[
w_0(X)
=
w^{\mathrm{raw}}(X;\beta_0,\pi),
\qquad
a_0=P[w_0(X)]
=
E[w_0(X)].
\]
By definition,
\[
\theta_{\mathrm{EW}}
=
\frac{P[w_0(X)Y]}{a_0}.
\]
Also define, for fixed $P$,
\[
T_P(\beta,\tilde\pi)
=
\frac{
P[
w^{\mathrm{raw}}(X;\beta,\tilde\pi)Y
]
}{
P[
w^{\mathrm{raw}}(X;\beta,\tilde\pi)
]
}.
\]
Thus $T_P(\beta_0,\pi)=\theta_{\mathrm{EW}}$, and the sample
estimator can be written as
\[
\hat\theta_{\mathrm{EW}}
=
\frac{
P_n[
w^{\mathrm{raw}}(X;\hat\beta,\hat\pi)Y
]
}{
P_n[
w^{\mathrm{raw}}(X;\hat\beta,\hat\pi)
]
}.
\]
We decompose its estimation error into the direct sampling
contribution, the effect of estimating the model parameters, and a
remainder arising from evaluating the empirical weighted mean at the
estimated arguments:
\begin{align}
\hat\theta_{\mathrm{EW}}-\theta_{\mathrm{EW}}
={}&
\left\{
\frac{P_n[w_0Y]}{P_n[w_0]}
-
\frac{P[w_0Y]}{P[w_0]}
\right\}
\nonumber\\
&+
\left\{
T_P(\hat\beta,\hat\pi)
-
T_P(\beta_0,\pi)
\right\}
+
R_n,
\label{eq:proof-decomposition}
\end{align}
where $R_n$ denotes the empirical-process remainder caused by
replacing $(\beta_0,\pi)$ with $(\hat\beta,\hat\pi)$ inside the
empirical weighted mean.

Consider first the direct sampling contribution. Since
\[
\frac{P_n[w_0Y]}{P_n[w_0]}
-
\frac{P[w_0Y]}{P[w_0]}
=
\frac{
(P_n-P)[w_0(Y-\theta_{\mathrm{EW}})]
}{
P_n[w_0]
},
\]
and $P_n[w_0]\xrightarrow{p}a_0$, the law of large numbers and the
central limit theorem yield
\[
\sqrt n
\left\{
\frac{P_n[w_0Y]}{P_n[w_0]}
-
\theta_{\mathrm{EW}}
\right\}
=
\frac1{\sqrt n}
\sum_{i=1}^n
\frac{w_0(X_i)}{a_0}
\{Y_i-\theta_{\mathrm{EW}}\}
+
o_p(1).
\]
Thus the direct sampling contribution to the influence function is
\[
\psi_{\mathrm{dir}}(Z)
=
\frac{
w^{\mathrm{raw}}(X;\beta_0,\pi)
}{
a_0
}
\{Y-\theta_{\mathrm{EW}}\}.
\]

Next, by Lemma~\ref{lem:weight-differentiability},
$T_P(\beta,\tilde\pi)$ is continuously differentiable. Define
\[
B_\beta
=
\left.
\frac{\partial T_P(\beta,\pi)}
{\partial\beta}
\right|_{\beta=\beta_0}.
\]
A first-order Taylor expansion, together with
$\hat\beta-\beta_0=O_p(n^{-1/2})$, gives
\[
T_P(\hat\beta,\pi)
-
T_P(\beta_0,\pi)
=
B_\beta^\top
(\hat\beta-\beta_0)
+
o_p(n^{-1/2}).
\]
Under the stated logistic regularity conditions,
\[
\sqrt n(\hat\beta-\beta_0)
=
\frac1{\sqrt n}
\sum_{i=1}^n
\mathcal I(\beta_0)^{-1}
X_i
\{Y_i-p_0(X_i)\}
+
o_p(1).
\]
Consequently,
\[
B_\beta^\top
\sqrt n(\hat\beta-\beta_0)
=
\frac1{\sqrt n}
\sum_{i=1}^n
B_\beta^\top
\mathcal I(\beta_0)^{-1}
X_i
\{Y_i-p_0(X_i)\}
+
o_p(1).
\]

The empirical cell proportions also contribute to the first-order
expansion. Specifically,
\[
\hat\pi_s
=
\frac1n
\sum_{i=1}^n
I(X_i=x_s).
\]
Define the cell-indicator vector
\[
e(X)
=
\left(
I(X=x_1),\ldots,I(X=x_S)
\right)^\top.
\]
Then $\hat\pi=P_n e(X)$, so
\[
\sqrt n(\hat\pi-\pi)
=
\frac1{\sqrt n}
\sum_{i=1}^n
\{e(X_i)-\pi\}.
\]
Writing
\[
B_\pi
=
\left.
\frac{\partial T_P(\beta_0,\tilde\pi)}
{\partial\tilde\pi}
\right|_{\tilde\pi=\pi},
\]
a first-order Taylor expansion gives
\[
T_P(\beta_0,\hat\pi)
-
T_P(\beta_0,\pi)
=
B_\pi^\top
(\hat\pi-\pi)
+
o_p(n^{-1/2}).
\]
It follows that
\[
B_\pi^\top
\sqrt n(\hat\pi-\pi)
=
\frac1{\sqrt n}
\sum_{i=1}^n
B_\pi^\top
\{e(X_i)-\pi\}
+
o_p(1).
\]

It remains to show that replacing $(\beta_0,\pi)$ with
$(\hat\beta,\hat\pi)$ inside the empirical weighted mean does not
produce an additional first-order contribution. By
Lemma~\ref{lem:weight-regularity}, the class
\[
\mathcal F
=
\left\{
w^{\mathrm{raw}}(\cdot;\beta,\tilde\pi)Y:
(\beta,\tilde\pi)\in\mathcal N\times\Pi
\right\}
\]
is uniformly bounded. Moreover, by
Lemma~\ref{lem:weight-differentiability}, the map
$(\beta,\tilde\pi)\mapsto
w^{\mathrm{raw}}(X;\beta,\tilde\pi)$ is uniformly Lipschitz on the
finite-dimensional compact parameter space
$\mathcal N\times\Pi$. Hence $\mathcal F$ is a finite-dimensional
parametric class with a bounded envelope and is $P$-Donsker.
Since $(\hat\beta,\hat\pi)\xrightarrow{p}(\beta_0,\pi)$,
stochastic equicontinuity gives
\[
\sqrt n(P_n-P)
\left[
w^{\mathrm{raw}}(X;\hat\beta,\hat\pi)Y
-
w^{\mathrm{raw}}(X;\beta_0,\pi)Y
\right]
=
o_p(1).
\]
The same argument applies to the denominator, and therefore
$R_n=o_p(n^{-1/2})$.

Combining the direct sampling contribution with the two plug-in
contributions and the negligible remainder yields
\[
\sqrt n
(\hat\theta_{\mathrm{EW}}-\theta_{\mathrm{EW}})
=
\frac1{\sqrt n}
\sum_{i=1}^n
\psi(Z_i)
+
o_p(1),
\]
where
\[
\psi(Z)
=
\frac{
w^{\mathrm{raw}}(X;\beta_0,\pi)
}{
a_0
}
\{Y-\theta_{\mathrm{EW}}\}
+
B_\beta^\top
\mathcal I(\beta_0)^{-1}
X
\{Y-p_0(X)\}
+
B_\pi^\top
\{e(X)-\pi\}.
\]
Each term has finite second moment under the stated assumptions.
Therefore, by the central limit theorem,
\[
\sqrt n
(\hat\theta_{\mathrm{EW}}-\theta_{\mathrm{EW}})
\Rightarrow
N(0,V),
\]
where
\[
V=E[\psi(Z)^2].
\]

Finally, the plug-in estimators
$\hat B_\beta$, $\hat B_\pi$, $\hat{\mathcal I}$, and the normalized
equivalent weights are consistent by the preceding regularity results
and the law of large numbers. Hence the estimated influence
contributions converge in $L_2$ to $\psi(Z)$, and consequently
\[
\widehat V
=
\frac1n
\sum_{i=1}^n
\hat\psi_i^2
\xrightarrow{p}
V.
\]
\end{proof}

\section{Data and Code Availability}

\doublespacing

Selected R code used to reproduce the simulations and empirical analyses in this paper is publicly available at
\url{https://github.com/huneel777/logit_equivalent_weight/tree/main/logit_weight_code}.

The empirical application uses data from the Future of Families and
Child Wellbeing Study (FFCWS). Public-use data and documentation are
available at
\href{https://ffcws.princeton.edu/documentation}
{https://ffcws.princeton.edu/documentation}.

\begingroup

\setlength{\bibsep}{3pt}

\endgroup

\end{document}